\documentclass[11pt]{article}

\usepackage{times}
\usepackage{amsmath}
\usepackage{algorithm}
\usepackage{amsthm}
\usepackage{graphicx,epsfig}
\usepackage{caption}
\usepackage[noend]{algpseudocode}

\theoremstyle{plain} \numberwithin{equation}{section}
\newtheorem{theorem}{Theorem}[section]
\newtheorem{corollary}[theorem]{Corollary}

\newtheorem{lemma}[theorem]{Lemma}

\theoremstyle{plain}
\newtheorem{invariant}[theorem]{Invariant}
\newtheorem{property}[theorem]{Property}

\newtheorem{claim}{Claim}[section]
\newcommand{\C}{\mathcal{C}}
\renewcommand{\P}{\mathcal{P}}

\newcommand{\E}{\mathcal{E}}

\advance \topmargin by -\headheight
\advance \topmargin by -\headsep
\evensidemargin \oddsidemargin
\usepackage{bbm}
\usepackage{xcolor}

\usepackage{enumitem}
\usepackage{xspace}
\renewcommand{\paragraph}[1]{\medskip\noindent{\bf #1}\xspace}
\usepackage{geometry}
\newtheorem{observation}[theorem]{Observation}
\newtheorem{algo}[theorem]{Algorithm}

\definecolor{ForestGreen}{rgb}{0.1333,0.5451,0.1333}
\definecolor{DarkRed}{rgb}{0.8,0,0}
\definecolor{Red}{rgb}{1,0,0}
\usepackage[linktocpage=true,
pagebackref=true,colorlinks,
linkcolor=DarkRed,citecolor=ForestGreen,
bookmarks,bookmarksopen,bookmarksnumbered]
{hyperref}
\renewcommand*\backref[1]{\ifx#1\relax \else (cit. on p. #1) \fi} 

\ifdefined\ShowComment

\def\danupon#1{\marginpar{$\leftarrow$\fbox{D}}\footnote{$\Rightarrow$~{\sf\textcolor{blue}{#1 --Danupon}}}}

\else
\def\danupon#1{}

\fi 

\begin{document}

\title{A New Deterministic  Algorithm for Dynamic Set Cover}

\author{Sayan Bhattacharya\thanks{University of Warwick, UK. Email: {\tt S.Bhattacharya@warwick.ac.uk}} \and Monika Henzinger\thanks{University of Vienna, Austria. Email: {\tt monika.henzinger@univie.ac.at}} \and Danupon Nanongkai\thanks{KTH Royal Institute of Technology, Sweden. Email: {\tt danupon@kth.se }}}

\date{}

\renewcommand{\L}{\mathcal{L}}


\renewcommand{\S}{\mathcal{S}}
\renewcommand{\E}{\mathcal{E}}
\newcommand{\cI}{\mathcal{I}}
\newcommand{\garbage}{\mathcal{G}}

\newcommand{\poly}{\operatorname{poly}}

\begin{titlepage}
\maketitle

\begin{abstract}
We present a deterministic dynamic algorithm for maintaining a $(1+\epsilon)f$-approximate  minimum cost set cover with $O(f\log(Cn)/\epsilon^2)$ amortized update time, when the input set system is undergoing element insertions and deletions. Here, $n$ denotes the number of elements, each element appears in at most $f$ sets, and the cost of each set lies in the range $[1/C, 1]$.  Our result, together with that  of Gupta~et~al.~[STOC`17], implies that there is a deterministic algorithm for this problem with $O(f\log(Cn))$ amortized update time and $O(\min(\log n, f))$-approximation ratio, which nearly matches the polynomial-time hardness of approximation for minimum set cover in the static setting.
Our update time is only $O(\log (Cn))$ away from a trivial lower bound.

Prior to our work, the previous best approximation ratio guaranteed by deterministic algorithms was $O(f^2)$, which was due to Bhattacharya~et~al.~[ICALP`15].  In contrast, the only result that guaranteed $O(f)$-approximation was obtained very  recently by Abboud~et~al.~[STOC`19], who designed a dynamic algorithm with  $(1+\epsilon)f$-approximation ratio and   $O(f^2 \log n/\epsilon)$ amortized update time. Besides the extra $O(f)$ factor in the update time compared to our and Gupta~et~al.'s results, the Abboud~et~al.~algorithm is {\em randomized}, and works only when the adversary is {\em oblivious} and the sets are unweighted (each set has the same cost).

We achieve our result via the primal-dual approach, by maintaining a fractional packing solution as a dual certificate. This approach was pursued  previously by Bhattacharya~et~al.~and Gupta~et~al., but not in the recent paper by Abboud~et~al. Unlike previous primal-dual algorithms that try to satisfy some  {\em local} constraints for individual sets at all time, our algorithm basically waits until the dual solution changes significantly {\em globally}, and fixes the solution only where the fix is needed.
\end{abstract}

	\newpage 
	
	\setcounter{tocdepth}{2}
	\tableofcontents

\end{titlepage}

\pagenumbering{arabic}

\newpage

\section{Introduction}
\label{sec:intro}

In the (static) set cover problem, an algorithm is given a collection $\S$ of $m$ sets over a universe $\E$ of $n$ elements such that $\cup_{s \in \S} s = \E$. Each set $s \in \S$ has a positive {\em cost}  $c_s$. After scaling these costs  by some appropriate factor, we can always get a parameter $C > 1$ such that:
\begin{equation}
\label{eq:weight}
 1/C \leq c_s \leq 1 \text{ for every set } s \in \S.
\end{equation}  
For any $\S' \subseteq \S$, let $c(\S') = \sum_{s \in \S'} c_s$ denote the total cost of all the sets in $s \in \S'$.  We say that a set $s \in \S$ {\em covers} an element $e \in \E$ iff $e \in s$. Our goal is to pick a collection of sets $\cI \subseteq \S$ with minimum total cost $c(\cI)$ so as to cover all the elements in the universe $\E$.

%

Set cover is a fundamental optimization problem that has been extensively studied in the contexts of polynomial-time approximation algorithms and online algorithms. 
In recent years, it has received significant attention in the  {\em dynamic algorithms} community as well, where the goal is to 
%
 maintain a set cover $\cal I\subseteq \S$ of small cost efficiently under a sequence of element insertions/deletions in $\E$. 
In particular, a dynamic algorithm for set cover must support the following {\em update operations}.


\danupon{The operations below are different from what we discussed last time. I used this because they are closer to previous works.}

\smallskip
\noindent {\sc Preprocess}($m$): Create $m$ empty sets in $\S$. Return $\cI = \emptyset$, and identifiers (e.g., integers) to the sets in $\S$. 

\smallskip
\noindent {\sc Insert}(${\cal F}=\{s_1, s_2, \ldots\}$): Insert to $\E$ a new element $e$ which belongs to the sets $s_1, s_2, \ldots$ (their identifiers are given as parameters). Return an identifier to the new element $e$, and the identifiers of sets that get added to and removed from $\cal I$. 

\smallskip
\noindent {\sc Delete}($e$): Delete element $e$ from $\E$.  Return the identifiers of sets that get added to and removed from $\cal I$.


\smallskip
After each update, the algorithm must guarantee that $\cI$ is a set cover; i.e. every element $e\in \E$ is in some set in $\cI$. Let $f$ and $n$ be the maximum size of $\cal F$ and $\E$, respectively, over all updates. The parameter $f$ is known as the {\em maximum frequency}.
%
%
%
%
It is usually assumed that $f, n$ and $m$ are known and fixed in the beginning, but note that our algorithm does not really need this assumption.\footnote{We mention that our algorithm do not really need the preprocessing step. When {\sc Insert}(${\cal F}=\{s_1, s_2, \ldots\}$) is called with a new set $s_i$, it can simply create $s_i$ on the fly.}
Note that dynamic set cover as defined above is a generalization of the dynamic vertex cover problem which, together with the dynamic maximum matching problem, have been studied extensively in recent years (e.g. \cite{OnakR10,BaswanaGS18,GuptaP13,BhattacharyaHN17,BhattacharyaCH17,BhattacharyaHI18-SICOMP-Matching,BhattacharyaHN16,Solomon16,NeimanS16,PelegS16,GLSSS19,BernsteinFH19,BernsteinS15,BernsteinS16,Sankowski07,BrandNS-FOCS19,AbboudW14,HenzingerKNS15}).


The performance of dynamic algorithms is mainly measured by the {\em update time}, the time to handle each  {\sc Insert} and  {\sc Delete} operation. Previous works on set cover focus on the {\em amortized update time}, where an algorithm is said to have an amortized update time of $t$ if, for any $k$, the total time it spends to process the first $k$ updates is at most $kt$. 
We also consider only the amortized update time in this paper, and simply use ``update time'' to refer to ``amortized update time''.  
The time for the {\sc Preprocess} operation is called the {\em preprocessing time}. It is typically not a big concern as long as it is polynomial.\footnote{Our algorithm requires only linear preprocessing time.}

\smallskip
 \noindent {\bf Perspective:}
Since the static set cover problem is NP-complete, it is natural to consider approximation algorithms. An algorithm has an approximation ratio of  $\alpha$ if outputs a set cover $\cI$ with $c(\cI)\leq \alpha \cdot OPT$, where $OPT$ is  cost of the optimal set cover.
%
Since the tight approximation factors for polynomial-time static set cover algorithms are $\Theta(\log n)$ and $\Theta(f)$  (e.g. \cite{DinurS14,DinurGKR05,Chvatal79,Slavik97,KhotR08}), it is natural to ask if one can also obtain these same guarantees in the dynamic setting with small update time. 
The $O(\log n)$ approximation ratio was already achieved in 2017 via greedy-like techniques by Gupta~et~al.~\cite{GuptaK0P17}. Their algorithm is deterministic and has $O(f\log n)$ update time. This update time is only $O(\log n)$ away from the trivial $\Omega(f)$ lower bound --  the time needed for specifying the sets that contain a given element (which is currently being inserted).
A similar lower bound holds even in some settings where updates can be specified with less than $f$ bits \cite{AbboudAGPS-stoc19}, e.g., when elements and sets are fixed in advance, and the updates are activations and deactivations of elements.\footnote{Abboud~et~al. \cite{AbboudAGPS-stoc19} showed that, under SETH, there is no algorithm with polynomial preprocessing time and $f^{1-\delta}$ update time for any constant $\delta>0$ when elements and sets are fixed in advance, and the updates are activations and deactivations of elements.}
%
%
The $O(f)$ approximation ratio  was recently achieved by Abboud~et~al.~\cite{AbboudAGPS-stoc19} (improving upon the approximation factors of $O(f^2)$ and higher by \cite{BhattacharyaHI18,GuptaK0P17,BhattacharyaCH17}).  Abboud~et~al.~show how to maintain an $(1+\epsilon)f$-approximation in $O(f^2\log n/\epsilon)$ amortized update time. Their algorithm, however, is randomized and does not work for the weighted case (when different sets have different costs).\footnote{A fundamental difficulty to extend Abboud~et~al.'s algorithm to the weighted case is the static algorithm it is based on. This static algorithm repeatedly picks an element $e$ that is not yet covered, and adds all sets $e$ belong to in the set cover solution. It is easy to prove that this algorithm returns an $f$-approximation in the unweighted case. It is also easy to construct an example that shows that this algorithm cannot guarantee any reasonable approximation ratio for the weighted case. } 
Like most randomized dynamic algorithms  currently existing in the literature, it works only when the future updates do not depend on the algorithm's past output -- this is also known as the {\em oblivious adversary assumption}. Removing this assumption is a central question in this area, since it may in general make many dynamic algorithms useful as subroutines inside fast static algorithms~\cite{ChuzhoyK19,NanongkaiSY-stoc19,BernsteinC-soda17,BernsteinC16,NanongkaiS17,NanongkaiSW17,BhattacharyaHN16,HenzingerKN16}. Accordingly, prior to our work,  it was natural to ask if there is an efficient $O(f)$-approximation algorithm for dynamic set cover that  is {\em deterministic} and/or can handle the {\em weighted case}.  In this paper we answer this question positively.

%
%


\begin{table*}
	\centering
	\begin{centering}
		\begin{tabular}{|c|c|c|c|c|}
			\hline 
			\textbf{Reference} & \textbf{Approximation Ratio} & \textbf{Update Time} & \textbf{Deterministic?} & \textbf{Weighted?} \tabularnewline
			\hline 
			\cite{GuptaK0P17} & $O(\log n)$	&	$O(f \log n)$ & yes & yes\\
			\cite{GuptaK0P17,BhattacharyaCH17} & $O(f^3)$	&	$O(f^2)$ & yes & yes\\
			\cite{BhattacharyaHI18} & $O(f^2)$ & $O(f \log(m + n))$ & yes & yes\\
			\cite{AbboudAGPS-stoc19} & $(1+\epsilon)f$ & $O(f^2 \log n/\epsilon)$ & no & no\\ 
			\textbf{Our result} &  $(1+\epsilon)f$ & $O(f \log (Cn)/\epsilon^{2})$ & {\bf yes} & {\bf yes}\\
			\hline
		\end{tabular}
	\end{centering}
\caption{Summary of results on dynamic set cover}
	\label{tab:1}
\end{table*}


\begin{theorem}
\label{thm:intro:main}
We can maintain a $(1+\epsilon)f$-approximate minimum-cost set cover in the dynamic setting, deterministically, with $O(f \log (Cn)/\epsilon^2)$ amortized update time, where $C$ is the ratio between the maximum and minimum set costs. 
\end{theorem}


Thus, we simultaneously  (a) improve upon the update time of Abboud~et~al.~\cite{AbboudAGPS-stoc19}, (b) derandomize their result, and (c) extend their result to the weighted case. 
%
Our algorithm, together with the one in Gupta~et~al.~\cite{GuptaK0P17}, settles an important open question in a line of work on dynamic set cover \cite{BhattacharyaCH17,BhattacharyaHI18,GuptaK0P17,AbboudAGPS-stoc19}.  We can now get an $O(\min(\log n, f))$-approximation using a deterministic algorithm with $O(f\log (Cn))$ update time. The approximation ratio matches the one achievable by the best possible polynomial-time static algorithm, whereas the update time is only  $O(\log (Cn))$  away from a trivial lower bound of $\Omega(f)$.


\subsection{Technical Overview}

%

\noindent 
{\bf Previous Approaches:} 
The {\em primal-dual schema} is  a powerful tool for designing many static approximation algorithms. In recent years, it has also been the main driving force behind {\em deterministic} dynamic algorithms for set cover and maximum matching (e.g. \cite{BhattacharyaHI18-SICOMP-Matching,BhattacharyaHN16,BhattacharyaCH17,BhattacharyaHI18,GuptaK0P17}), including all $O(\poly(f))$-approximations for dynamic set cover {\em except } the one by Abboud~et~al.~\cite{AbboudAGPS-stoc19}. 
%
%
%

The {\em dual} of minimum set cover  happens to be a {\em fractional packing} problem, which is defined as follows. Given a set system $(\S, \E)$ as input, we have to assign a fractional weight $w(e) \geq 0$ to every element. We want to maximize  $\sum_{e \in \E} w(e)$, subject to the following constraint:
\begin{equation}
\label{eq:dual} 
\sum_{e \in s} w(e) \leq c_s \text{ for every set } s \in \S.
\end{equation} 
For the rest of this paper, we let $W(s) = \sum_{e \in s} w(e)$ denote the total weight received by a set $s \in \S$ from all its elements. Furthermore, define $w(S) = \sum_{e \in S} w(e)$ for every subset of elements $S \subseteq \E$. Thus, the goal is to maximize $w(\E)$ subject to the constraint that $W(s) \leq c_s$ for all sets $s \in \S$. 

%
Let $OPT(\S, \E)$ denote the total cost of the minimum set cover in $(\S, \E)$. We simply use $OPT$ when $(\S, \E)$ is clear from the context.
All the previous primal-dual algorithms for  dynamic set cover  try to maintain some invariants about {\em individual} sets all the time. In particular, they are based on the following lemma.  

\begin{lemma}\label{thm:intro:dual_lowerbound}
	Consider any fractional packing $w$ in $(\S, \E)$ that satisfy~\eqref{eq:dual}. Then we have $w(\E) \leq OPT(\S, \E)$. In addition, if there exist a set-cover $\cI\subseteq \S$ of $\E$ and an $\alpha\geq 1$ such that
	\begin{align}
 W(s)\geq c_s/\alpha \text{ for all sets } s \in \cI, \medskip\label{eq:intro:invariant}
	\end{align}
	then we have $c(\cI)\leq  \alpha f \cdot w(\E) \leq \alpha f \cdot OPT(\S, \E)$.
\end{lemma}

%

All the previous dynamic primal-dual algorithms maintain a set cover $\cI$  and a fractional packing $w$  that satisfy~\eqref{eq:intro:invariant}. Needless to say, the algorithms have to  change  $\cI$ and the weights $w(e)$ of some elements $e$ in a carefully chosen manner after each update. 
For example, the algorithm of Bhattacharya~et~al.~\cite{BhattacharyaHI18} satisfies~\eqref{eq:intro:invariant} with $\alpha=O(f)$, implying an $O(f^2)$-approximation factor. 
To obtain an $O(f)$ approximation factor we need to satisfy \eqref{eq:intro:invariant} with $\alpha=O(1)$. 
However, as pointed out by Abboud~et~al. \cite{AbboudAGPS-stoc19}, it is not clear how to maintain such a strict constraint efficiently for $\alpha = O(1)$. 
The trouble is that one update may violate \eqref{eq:intro:invariant} and may cause a sequence of weight changes for many elements. This creates difficulties for bounding the update time (which typically require intricate arguments via clever potential functions).  
%
%
%

Because of this difficulty, Abboud~et~al.~opted for a different approach that is based on the following static algorithm. (i) Pick any uncovered element $e$ uniformly at random, called {\em pivot}. (ii) Include all sets containing $e$ in the set cover solution. (iii) Repeat this process until all elements are covered. 
%
%
It is easy to see that this algorithm returns an $f$-approximation for the unweighted case (when every set has the same cost). 
%
Since this approximation ratio does not hold for the weighted case in the static setting, it seems difficult to extend the approach of Abboud~et~al.~to the weighted case. 
More importantly,  in the analysis of their algorithm in the dynamic setting, Abboud~et~al.~crucially rely on  randomness and the oblivious adversary assumption. This allows them to argue that before a pivot element gets deleted, many other non-pivot elements must  also get deleted in expectation. Thus they can charge the time their algorithm needs in handling the deletion of a pivot to the (large number of) non-pivot elements that got deleted in the past. 
This type of argument was also used for maintaining a maximal matching \cite{BaswanaGS18,Solomon16}. To the best of our knowledge, there was no technique to derandomize this type of argument. In fact, all the known deterministic dynamic algorithms for set cover and matching use the primal-dual schema. This leads to a basic question: {\em Is the primal-dual approach powerful enough to give an $O(f)$-approximation algorithm for dynamic set cover?}

We answer this question in the affirmative. Unlike  previous dynamic primal-dual algorithms, which try to satisfy conditions like \eqref{eq:intro:invariant} that are {\em local} to individual sets, our algorithm basically waits until the dual solution changes significantly {\em globally}, and then it fixes the solution only where the fix is needed.



%
%

\medskip
\noindent {\bf Our Approach and the Showcase (Batch Deletion):} 
To appreciate our main idea, consider the {\em batch deletion} setting, where we have to preprocess a set system $(\S, \E)$ and then there is an update $D\subseteq \E$ that changes the set system to $(\S, \E')$, where $\E'=\E\setminus D$. Our goal is to recompute an approximately minimum set cover in the new input $(\S, \E')$ in time proportional to the size of the update (i.e., $|D|$).

Suppose that originally we have a pair $(\cI, w)$ that satisfies \eqref{eq:intro:invariant} with $\alpha=1$ for $(\S, \E)$; thus, $c(\cI)\leq  f \cdot w(\E) \leq f \cdot OPT(\S, \E)$. Clearly,  $\cI$ remains a set cover of the new set system $(\S, \E')$. To simplify things even further, suppose that the element-weights are {\em uniform},\footnote{This means that $w(e)=\delta$ for every $e \in \E$ (for some $\delta\geq 0$).} and  $|D|\leq \epsilon \cdot |\E'|$ for some $\epsilon > 0$. This implies that  $w(D) \leq \epsilon \cdot w(\E')$. So Lemma~\ref{thm:intro:dual_lowerbound} gives us: 
\begin{eqnarray*}
c(\cI)\leq  f \cdot w(\E) & = & f \cdot \left(w(\E') + w(D) \right) \\
& \leq & f (1+\epsilon) \cdot w(\E') \\
& \leq &  f (1+\epsilon) \cdot OPT(\S, \E').
\end{eqnarray*}
\begin{observation}\label{obs:intro:uniform_count}
If  $|D|\leq \epsilon \cdot |\E'|$ and the element-weights are uniform, then   $c(\cI)\leq f (1+\epsilon) \cdot OPT(\S, \E').$
%
\end{observation}

In words, $\cI$ remains a good approximation to $OPT(\S, \E')$ if $|D|$ is small. Thus,  intuitively we do not need  to do anything when $|D|\leq \epsilon \cdot |\E'|$. This is already different from previous primal-dual algorithms that might have to do a lot of work, since \eqref{eq:intro:invariant} might be violated.
In contrast, when $|D|$ becomes larger than $\epsilon \cdot |\E'|$, in $O(f \cdot |\E'|)$ time  we can just compute a new pair $(\cI', w')$  satisfying~\eqref{eq:intro:invariant} with $\alpha=1$ for the set system $(\S, \E')$. 
Since we do this only after $\epsilon \cdot |\E'|$ deletions, we get an amortized update time of $O(f/\epsilon)$. 

One lesson from the above argument is this: Instead of trying to satisfy \eqref{eq:intro:invariant}, we might benefit 
from dealing with $D$ only when it is large enough, for this might help us ensure that the amortized update time remains small. 
Of course this is easy to argue under the uniform weight assumption,  which often makes the situation too simple.
%
%
The following lemma is the key towards doing something similar in the general setting. 


\begin{lemma}\label{lem:intro:new_invariant}
	Define  $s_{> x}=\{e\in s: w(e)> x\}$ for every set $s \subseteq \E$ of elements. Suppose that:  
	\begin{align}
   |D_{> x}| \leq  \epsilon \cdot |\E'_{> x}|  \ \ \text{ for all } x \geq 0. \label{eq:intro:new_invariant}
	\end{align}
Then  we have  $w(D)\leq \epsilon \cdot w(\E')$, and thus $c(\cI)\leq f (1+\epsilon) \cdot OPT(\S, \E').$
\end{lemma}
\begin{proof}[Proof sketch]
Note that  $w(e) = \int_{0}^{\infty} \mathbbm{1}(x< w(e)) dx$, where $\mathbbm{1}(x< w(e))$ is one if $x< w(e)$, and zero otherwise. Thus, we have:
\begin{eqnarray*}
 \sum_{e \in D} w(e)  & = & \int_{0}^{\infty} \sum_{e \in D} \mathbbm{1}(x< w(e)) dx \\
 & = & \int_{0}^{\infty} |\{e\in D: w(e)> x\}| dx \\
 & = & \int_{0}^{\infty} |D_{> x}| dx  \\
& \leq & \epsilon \cdot \int_{0}^{\infty} |\E'_{> x}| dx \\
& = & \epsilon \cdot \int_{0}^{\infty} |\{e\in \E': w(e)> x\}| dx \\
& = & \epsilon \cdot \sum_{e \in \E'} w(e). 
\end{eqnarray*} 
%
\end{proof}


Lemma~\ref{lem:intro:new_invariant} tells us that if  $|D_{> x}|\leq \epsilon \cdot |\E'_{> x}|$ for all $x$, then we do not have to do anything. 
When $|D_{> x}|>\epsilon |\E'_{> x}|$ for some $x$, we need to ``fix'' $\cI$. We do so by running a static algorithm on some sets and elements. The static algorithm is described in Algorithm~\ref{alg:intro:static}.  We describe how to use it in Algorithm~\ref{alg:intro:batch_deletion}. 


\begin{algo}[Static uniform-increment]\label{alg:intro:static}
	Given an input $(\hat\S, \hat\E)$,  start with $\hat w(e) \leftarrow 0$ for all $e \in \hat \E$ and $\hat\cI \leftarrow \emptyset$. Repeat the following until $\hat \cI$ covers all elements:
	(i) Raise weights $\hat w(e)$ at the same rate for every $e\in  \hat \E$ not covered by $\hat \cI$, until some sets $s$ in $\hat \S\setminus \hat \cI$ are tight (i.e. $W(s)=c_s$). (ii) Add such sets to $\hat \cI$. 
\end{algo}

\begin{algo}[Batch deletion algorithm]\label{alg:intro:batch_deletion}
	Initially, compute $(\cI, w)$ by running Algorithm~\ref{alg:intro:static} on $(\S,\E)$. To handle $D$, let $x^*$ be the minimum $x$ such that \eqref{eq:intro:new_invariant} is violated and do the following. (Do nothing if \eqref{eq:intro:new_invariant} is not violated at all.) Run Algorithm~\ref{alg:intro:static} to compute $(\hat{\cI}, \hat{w})$ for  $(\S_{> x^*}, \E'_{> x^*})$, where $\S_{> x^*}=\{s\in \S : s\cap \E_{> x^*} \neq \emptyset\}$. Set $\cI\gets (\cI\setminus \S_{> x^*})\cup \hat{\cI}$, $D=D\setminus D_{> x^*}$, and $w(e)=\hat{w}(e)$ for all $e\in \E_{> x^*}$. 	
\end{algo}

In short, Algorithm~\ref{alg:intro:batch_deletion} fixes $\cI$ by running  the static Algorithm~\ref{alg:intro:static} on  input $(\S_{> x^*}, \E'_{> x^*})$. 
We can implement Algorithm~\ref{alg:intro:static} in $O(|\S_{> x^*}|+ |\E'_{> x^*}|)=O(f|\E_{> x^*}|)=O(f|D_{> x^*}|/\epsilon)$ time approximately.\footnote{We can implement an approximate version of Algorithm~\ref{alg:intro:static}, where element weights are in the form $(1+\epsilon)^i$ for some $i$, and we say that a set $s$ is tight if $W(s)\geq c_s/(1+\epsilon)$, increasing the approximation ratio by another $(1+\epsilon)$ multiplicative factor. It is not hard to see that this can be done in $O((|\S_{> x^*}|+ |\E_{> x^*}|)\log(n)/\epsilon)$ time. We further show that the $\log n$ term can be eliminated.}
This gives an amortized update time of $O(f/\epsilon)$, since we can {\em charge} the $O(f|D_{> x^*}|/\epsilon)$ time spent  in ``fixing" $\cI$ to the elements in $D_{> x^*}$ that got deleted.
%
The lemma below implies the correctness of this strategy. 

\begin{lemma}\label{lem:intro:correctness}
(i) Let $(\cI, w)$ denote the output of Algorithm~\ref{alg:intro:static} when given $(\S, \E)$ as input. For every $x\geq 0$,  the subset of elements $\E\setminus \E_{> x}$ is covered by $\cI\setminus \S_{> x}$.
%
(ii) After processing $D$, Algorithm~\ref{alg:intro:batch_deletion}  produces $D$, $\cI$, and $w$ that satisfy \eqref{eq:intro:new_invariant}; which implies that $c(\cI)\leq (1+\epsilon) f \cdot OPT(\E')$.\danupon{It will be great if somebody can confirm that this is true and if we can point to a lemma similar to this claim in the full body.}
\end{lemma}
\begin{proof}[Proof Idea]
For (i), Algorithm~(\ref{alg:intro:static}) stops raising the weight of some element $e\in \E\setminus \E_{> x}$ only when some set $s$ containing $e$ is tight. At this point we also stop raising the weight of every element in $s$, making $s\in \cI\setminus \S_{> x}$.
For (ii), an intuition is that Algorithm~\ref{alg:intro:batch_deletion} has subtracted from $D$ all $D_{>x}$ that violate \eqref{eq:intro:new_invariant}. This subtraction does not increase $D_{>x}$ for any $x$, and, thus, does not create any new violation to  \eqref{eq:intro:new_invariant}. 
\end{proof}

Lemma~\ref{lem:intro:correctness}(i) implies that we can add/remove to/from $\cI$ sets in $\S_{> x^*}$ without worrying about the coverage of elements in $\E\setminus \E_{> x^*}$ (they will be covered even when we remove all sets in $\S_{> x^*}$ from $\cI$). Consequently, we can guarantee that after processing $D$, Algorithm~\ref{alg:intro:batch_deletion} produces a set $\cI$ that covers all elements. 
Lemma~\ref{lem:intro:correctness}(ii) immediately implies the claimed approximation guarantee. Note that it is crucial to apply our new Lemma~\ref{lem:intro:new_invariant} with appropriate $D$, $\cI$, and $w$, which are changed after Algorithm~\ref{alg:intro:batch_deletion} processes $D$.


Note that running Algorithm~\ref{alg:intro:static} at the preprocessing is crucial for the correctness of Algorithm~\ref{alg:intro:batch_deletion}, as otherwise $\cI$ might not cover all elements after Algorithm~\ref{alg:intro:batch_deletion} finishes. In other words, Lemma~\ref{lem:intro:correctness}(i) might not be true if we replace Algorithm~\ref{alg:intro:static}  by some other static algorithm.\footnote{Consider, e.g., when $\E=\{e_1, \ldots, e_{12}\}\mbox{, } \S=\{s_1=\{e_1, \ldots, e_{10}\}, s_2=\{e_{10}, e_{11},e_{12}\}\} \mbox{, } c_{s_1}=c_{s_2}=10 \mbox{, and } D=\{e_{12}\}.$ If we start with all element-weights being zero except $w(e_{1})=w(e_{12})=10$, a deletion of $e_{12}$ causes $D_{>0}=\{e_{12}\}$ and $\E'_{>0}=\{e_{1}\}$, making the new set system to violate \eqref{eq:intro:new_invariant} at $x=0$. But it is not enough to change only the weight of $e_{1}$ which is the only element in $\E'_{>0}$.  Lemma~\ref{lem:intro:correctness}(i) guarantees that this will not happen if  $(\cI, w)$ is computed in a certain way as in Algorithm~\ref{alg:intro:static}.
}
\danupon{To save space, we can remove everything after this in this paragraph.}
We believe that this is key that gives the running time improvement over Abboud et al.'s algorithm, since otherwise we may have to spend more time checking if other elements remain covered. (This is essentially what happened in Abboud~et~al.'s algorithm.)\danupon{Is this correct? This is my impression from talking to Monika.}
Note that Algorithm~\ref{alg:intro:static} was also used in previous dynamic algorithms \cite{BhattacharyaCH17,BhattacharyaHI18,GuptaK0P17}, but to our knowledge it does not play a role in the correctness of those algorithms like in our algorithm.\danupon{Correct for Gupta et al.?}

One detail to mention is how Algorithm~\ref{alg:intro:batch_deletion} finds $x^*$. One simple way is to round element weights to the form $(1+\epsilon)^i$ for integers $i$. (We refer to such $i$ as the {\em level} of an element in the rest of this paper.) With this rounding, we simply have to search for $O(\log(Cn))$ different choices of $x^*$, where $C$ is defined in Theorem~\ref{thm:intro:main}.
The total update time amortized over deletions in $D$ becomes $O(\log(Cn)+f/\epsilon)$. 

\danupon{If we have time, it might be nice to try to eliminate the log term for the decremental case.}


	
\medskip\noindent{\bf The Fully-Dynamic Algorithm (Sketched).} Extending the above algorithm to handle more deletions is rather straightforward: We include a newly deleted element to $D$, check for $x^*$, and then fix  $(\S_{>x^*}, \E'_{>x^*})$ as in Algorithm~\ref{alg:intro:batch_deletion} if such a $x^*$ exists. This gives a decremental (i.e. deletion-only) algorithm with  $O(\log(Cn)+f/\epsilon)$ update time. This bound is faster than the $O(f^2/\epsilon^5)$ bound from \cite{AbboudAGPS-stoc19} when $f=\omega(\sqrt{\log n})$ and might be of an independent interest given that decremental algorithms have been heavily studied and lead to some applications (e.g. \cite{ChuzhoyK19,BernsteinC16,BernsteinC-soda17,HenzingerKN18,HenzingerKN-FOCS14,HenzingerKN-STOC14}).

Handling insertions, on the other hand, is more intricate. When an element $e$ is inserted, we set its weight $w(e)$ to the maximum possible value to make some set $s$ containing $e$ tight, i.e. $\sum_{e\in s \cap(\E\cup D)} w(e)= c_s$ (note that elements in $D$ also contribute to the weights of sets). This means that if $e$ is already in a tight set, then $w(e)=0$. Otherwise, it is increased until a new tight set $s$ is created, which will be added to $\cI$. We keep the newly inserted elements in a separate set (call it $\E'$ for now) because they do not get weights in the uniform way (like when we run Algorithm~\ref{alg:intro:static}).\footnote{In particular, we can show that weights of all elements in the set system $(\S, \E\cup D)$ are as if we run the static algorithm (Algorithm~\ref{alg:intro:static}) on this set system. We cannot say the same for $(\S, \E\cup D\cup \E')$.}
When Algorithm~\ref{alg:intro:batch_deletion} calls Algorithm~\ref{alg:intro:static} on some $(\S_{>x^*}, \E'_{>X^*})$, it will try to include in a greedy manner elements 
from $\E'$ in the uniform weight increment process and move them to $\E$. See Section~\ref{sec:dynamic} and Appendix~\ref{app:rebuild} for the details.

\danupon{To do: point to relevant sections. Mention algorithm names (e.g. Rebuild, fix levels)}
\danupon{I didn't get to say that how Abboud et al and our algorithms work on each level looks similar (but how we get levels are completely different), because I can't find a good place to discuss this and it's not clear that it's needed.}

\section{Minimum Set Cover in the Static Setting}
\label{sec:preprocessing}

In this section, we describe some basic concepts about the set cover problem {\em in the static setting}. We use the notations that were introduced in Section~\ref{sec:intro}. We start with a simple lemma that follows from LP-duality.

\begin{lemma}
\label{lm:duality}
Consider a valid set cover $\S' \subseteq \S$ and an assignment of nonnegative weights $\{w(e)\}$ to every element $e \in \E$ that satisfy~(\ref{eq:dual}), i.e.~$\{w(e)\}$ forms a valid fractional packing.  If $c(\S') \leq \alpha \cdot \sum_{e \in \E} w(e)$, then $\S'$ is an $\alpha$-approximate  minimum set cover.
\end{lemma}

\newcommand{\X}{\mathcal{X}}
\newcommand{\Y}{\mathcal{Y}}

 In Section~\ref{sec:intro}, we described  a simple static primal-dual algorithm that returns an $f$-approximate minimum set cover (see Algorithm~\ref{alg:intro:static}).  We  now consider a {\em discretized} variant of the above algorithm, which  increases the weights of elements in powers of $(1+\epsilon)$, instead of increasing these weights in a continuous manner.  This  results in a {\em hierarchical partition} of the set-system $(\S, \E)$, which assigns the sets and elements to different levels. In the Appendix, we explain  how the algorithm generates this hierarchical partition. Here, we only state some important properties of the partition and show how these properties imply a $(1+\epsilon)f$-approximation for the minimum set cover problem. For the rest of the paper, we fix two parameters  $\epsilon, L$.  
\begin{equation}
\label{eq:L}
0 < \epsilon < 1/2 \ \text{ and } \ L =  \lceil \log_{(1+\epsilon)} (C \cdot n) \rceil + 1.
\end{equation}


The  algorithm outputs a {\em hierarchical partition} of the set-system $(\S, \E)$, where each set $s \in \S$ is assigned to some {\em level} $\ell(s) \in \{0, \ldots, L\}$. The level of an element $e \in \E$ is defined as the maximum level among all the sets it belongs to, i.e., $\ell(e) = \max\{ \ell(s) : s \in \S, e \in s\}$. Note that if $e \in s$, then  $\ell(e) \geq \ell(s)$.

\medskip
\noindent 
{\bf Tight and slack sets:} Recall that $W(s) = \sum_{e \in s} w(e)$ denotes the total weight received by a set $s \in \S$ from all its elements. We say that a set $s \in \S$ is {\em tight} if $(1+\epsilon)^{-1} c_s \leq W(s) \leq c_s$ and {\em slack} if $0 \leq W(s) < (1+\epsilon)^{-1} c_s$.  The  hierarchical partition returned by the algorithm satisfies the following properties.

\begin{property}
\label{prop:element:weight:level}
For every element $e \in \E$, we have $w(e) = (1+\epsilon)^{-\ell(e)}$.
\end{property}

\begin{property}
\label{prop:set:weight:level}
Every set $s \in \S$ has  $0 \leq W(s) \leq  c_s$. Furthermore, every set $s \in \S$ that is slack has $\ell(s) = 0$. 
\end{property}

\begin{property}
\label{prop:element:level}
Every element $e \in \E$ is contained in at least one tight set. 
\end{property}

\begin{lemma}
\label{lm:static:approx}
Let $\S_{tight} = \{s \in \S : (1+\epsilon)^{-1} c_s \leq W(s) \leq c_s\}$ denote the collection of tight sets. They form a $(1+\epsilon)f$-approximate minimum set cover of the input $(\S, \E)$.
\end{lemma}

\begin{proof}
Since each element belongs to at most $f$ sets, a simple counting argument gives us:
\begin{eqnarray}
(1+\epsilon)f \cdot \sum_{e \in \E} w(e) & \geq & (1+\epsilon) \cdot \sum_{s \in \S} W(s) \nonumber \\
& \geq & \sum_{s \in \S_{tight}} (1+\epsilon) \cdot W(s) \nonumber \\
& \geq & \sum_{s \in \S_{tight}}  c_s \nonumber \\
& = & c\left(\S_{tight}\right) \label{eq:static:approx}
\end{eqnarray}
By Property~\ref{prop:element:level}, every element $e \in \E$ is covered by some set in $\S_{tight}$. In other words, the sets in $\S_{tight}$ form a valid set cover. Furthermore, by Property~\ref{prop:set:weight:level}, we have $0 \leq W(s) \leq c_s$ for all sets $s \in \S$. Accordingly, the weights $\{w(e)\}$ assigned to the elements form a valid fractional packing. From~(\ref{eq:static:approx}) and Lemma~\ref{lm:duality}, we now infer that the sets in $S_{tight}$ form a $(1+\epsilon)f$-approximate minimum set cover of the input $(\S, \E)$. 
\end{proof}

\section{Our Dynamic Algorithm}
\label{sec:dynamic}

Consider the minimum set cover problem in a dynamic setting,  where the input  $(\S,\E)$ keeps changing via a sequence of element insertions and deletions. Specifically, during each update, an element is either inserted into or deleted from the set system $(\S, \E)$. When an element $e$ is inserted, we get to know about the  sets in $\S$ that contain the element $e$. We assume that $f$ remains an upper bound on the maximum frequency of an element throughout this sequence of updates (although our dynamic algorithm does not need to know the value of $f$ in advance).  We will present a {\em deterministic} dynamic algorithm for maintaining a $(1+\epsilon)f$-approximate minimum set cover in this setting with $O(f \cdot \log (Cn)/\epsilon^2)$ amortized update time. 

\subsection{Classification of elements} 
\label{sub:sec:classification}

The main idea behind our dynamic algorithm is  simple. We maintain a  relaxed version of the hierarchical partition from Section~\ref{sec:preprocessing} in a {\em lazy manner}. To be more specific, in the preprocessing phase we start with a set-system $(\S, \E)$ where $\E = \emptyset$. At this point, every set $s \in \S$ is at  level $\ell(s) = 0$ and has a weight $W(s) = 0$, and Properties~\ref{prop:element:weight:level},~\ref{prop:set:weight:level},~\ref{prop:element:level} are vacuously true.  Subsequently, while handling the sequence of updates, whenever we observe that a significant fraction of  elements has been deleted from  levels $\leq i$ for some $i \in [0, L]$, we {\em rebuild} all the levels $\{0, \ldots, i\}$ in a certain natural manner. We refer to the subroutine which performs this rebuilding as {\sc Rebuild}($\leq i$).

\medskip
We will classify elements into three distinct types -- {\em active, passive} and {\em dead}. Let $A, P$ and $D$ respectively denote the set of active, passive and dead elements. Informally, every element is {\em active} in the hierarchical partition described in Section~\ref{sec:preprocessing}, where we considered the static setting. To get the main intuition in the dynamic setting, consider an update at some time-step $t$, and suppose that this update does {\em not} lead to a call to the subroutine   $\text{{\sc Rebuild}}(\leq i)$ for any $i \in [0, L]$. Recall that a set $s \in \S$ is called {\em tight} when its weight lies in the range $[(1+\epsilon)^{-1}c_s, c_s]$ and {\em slack}  when its weight lies in the range $[0, (1+\epsilon)^{-1}c_s)$.   As in Section~\ref{sec:preprocessing}, suppose that the tight sets  in the hierarchical partition form a valid set cover just before the update at time-step $t$ (see Lemma~\ref{lm:static:approx}). Now, consider three possible cases.

\medskip
\noindent {\em Case (a): The update at time-step $t$ deletes an element $e$.} In this case, we classify the element $e$ as {\em dead}. We continue to {\em pretend}, however, that the element $e$ still exists and do {\em not} change its weight $w(e)$. Thus, we take the value of $w(e)$ into account while calculating the weight of any set  in the fractional packing solution.  This ensures that  the collection of  tight sets  remains a valid set cover for the current input $(\S, \E)$. 

\medskip
\noindent {\em Case (b): The update at time-step $t$ inserts an element $e$ that belongs to at least one tight set.}  In this case, we assign the element $e$ to level $\ell(e) = \max\{ \ell(s) : s \in \S, e \in s\}$, classify it as {\em passive}, and assign it a weight $w(e) = 0$. This ensures that the  tight sets  continue to remain a  set cover in  $(\S, \E)$.

\medskip
\noindent {\em Case (c): The update at time-step $t$ inserts an element $e$ such that all  sets containing $e$ are slack.} In this case, Property~\ref{prop:set:weight:level} implies that every set containing the element $e$ lies at level $0$. Hence, we assign the element $e$ also to level $\ell(e) = \max\{ \ell(s) : s \in \S, e \in s\} = 0$. Unlike in Case (b), however, here we can no longer leave the hierarchical partition unchanged, since in that event the collection of  tight sets   will no longer form a valid set cover. We address this issue in the following manner. Let $\S_e$ denote the collection of sets containing $e$. Note that $\left| \S_e \right| \leq f$. Let $\lambda > 0$ be the {\em  minimum value} such that if we increase the weight of each set in $\S_e$ by an additive $\lambda$, then the weight of some set  $s \in \S_e$ becomes equal to $c_s$. We classify the element $e$ as {\em passive}, and assign it a weight $w(e) = \lambda$. This ensures that  now the collection of tight sets again forms a valid set cover. This also leads to a  very important consequence, which is stated below.

\begin{claim}
\label{cl:residual:weight}
A passive element $e$ receives a weight of $w(e) \leq (1+\epsilon)^{-\ell(e)}$ just after getting inserted.
\end{claim}

\begin{proof}
In Case (b)  above, a passive element receives zero weight and  the claim trivially holds. For the rest of the proof, consider the scenario described  in Case (c) above. Recall that as per~(\ref{eq:weight}) we have $c_s \leq 1$ for every set $s \in \S$. Let $s' \in \S$ be a set containing $e$ whose weight becomes equal to $c_{s'}$ when we assign a weight of $\lambda$ to the element $e$ (see the description for Case (c) above). Thus, we must have $\lambda \leq c_{s'} \leq 1 = (1+\epsilon)^{0} = (1+\epsilon)^{\ell(e)}$.
\end{proof}


\subsection{Levels and Weights of elements}
\label{sub:sec:levels}
Throughout the duration of our algorithm, the level of an element $e$ (regardless of whether it is active, passive or dead) will be defined to be $\ell(e) = \max \{ \ell(s) : s \in \S, e \in s\}$. From the preceding discussion, we also conclude that the weights assigned to the elements satisfy the following conditions.

\smallskip
If an element $e$ is active, then   $w(e) = (1+\epsilon)^{-\ell(e)}$. In contrast, if an element $e$ is passive, then   $w(e) \leq  (1+\epsilon)^{-\ell(e)}$. Finally, if an element $e$ is dead, then  its weight  depends on its state at the time of its deletion. Specifically, if it was active at the time of its deletion, then $w(e) = (1+\epsilon)^{-\ell(e)}$. If it was passive at the time of its deletion, then $w(e) \leq (1+\epsilon)^{-\ell(e)}$. To summarize,  a dead element $e$ always has $w(e) \leq (1+\epsilon)^{-\ell(e)}$.

\subsection{The shadow input and the invariants} 
\label{sub:sec:shadow}

Recall that the set $\E$ is partitioned into two subsets, namely $A \subseteq \E$ and $P = \E \setminus A$. From the way we assign the weights to  elements, it follows  that our algorithm works by {\em pretending} as if  the dead elements {\em were} still present in the input. Accordingly, we consider an input $(\S, \E^*)$, where $\E^* = \E \cup D$. We refer to $(\S, \E^*)$   as the {\em shadow input} (as opposed to the actual input $(\S, \E)$). Indeed,  the hierarchical partition maintained by our dynamic algorithm will be similar to the one from Section~\ref{sec:preprocessing} on the shadow input $(\S, \E^*)$, barring the fact that the passive/dead elements  will have weights $w(e) \leq (1+\epsilon)^{-\ell(e)}$. To explain this more formally, we use the following notations. For every set $s \in \S$, we let $W(s) = \sum_{e \in s \cap \E} w(e)$ and $W^*(s) = \sum_{e \in s \cap \E^*} w(e)$ respectively denote the total weight of all the elements that belong to $s$ in  $(\S, \E)$ and in $(\S, \E^*)$. Our dynamic algorithm will satisfy the three invariants stated below. Invariant~\ref{inv:element:weight:level:shadow} follows from the discussion in Section~\ref{sub:sec:levels}. Invariant~\ref{inv:node:weight:shadow} is analogous to Property~\ref{prop:set:weight:level}, whereas  Invariant~\ref{inv:hyper-edge:shadow}  is analogous to Property~\ref{prop:element:level}.

\begin{invariant}
\label{inv:element:weight:level:shadow}
Consider any element $e \in \E^* = A \cup P \cup D$. The level of $e$  is defined as $\ell(e) = \max\{ \ell(s) : s \in \S, e \in s \}$.  If $e \in A$, then we have $w(e) = (1+\epsilon)^{-\ell(e)}$. Otherwise, if $e \in P \cup D$, then we have $0 \le w(e) \leq (1+\epsilon)^{-\ell(e)}$.
\end{invariant}

\begin{invariant}
\label{inv:set:weight:level:shadow}
\label{inv:node:weight:shadow}
Every set $s \in \S$ satisfies $0 \leq W^*(s) \leq c_s$. Furthermore, every set $s \in \S$ with weight $W^*(s) < (1+\epsilon)^{-1} c_s$ is at level $\ell(s) = 0$. 
\end{invariant}

\begin{invariant}
\label{inv:element:level:shadow}
\label{inv:hyper-edge:shadow}
Each element  $e \in \E \cup D$ is contained in at least one set $s \in \S$  with $(1+\epsilon)^{-1} c_s \leq W^*(s) \leq c_s$. 
\end{invariant}

Let $\S^*_{tight} \subseteq \S$ be the collection of sets with weights $(1+\epsilon)^{-1} c_s \leq W^*(s) \leq c_s$ in the hierarchical partition maintained by our algorithm. Replacing Properties~\ref{prop:set:weight:level},~\ref{prop:element:level} by Invariants~\ref{inv:node:weight:shadow},~\ref{inv:hyper-edge:shadow} in the proof of Lemma~\ref{lm:static:approx}, we conclude that  $\S^*_{tight}$ gives a $(1+\epsilon)f$-approximate minimum set cover in the shadow input $(\S, \E^*)$. Invariant~\ref{inv:hyper-edge:shadow} further implies that $\S^*_{tight}$ is a valid set cover in the actual input $(\S, \E)$. We will show later  that $\S^*_{tight}$ is in fact a $(1+O(\epsilon))f$-approximate minimum set cover in the actual input $(\S, \E)$ as well. This happens because, intuitively, our dynamic algorithm ensures that the actual input $(\S, \E)$ always remains  {\em close} to the shadow input $(\S, \E^*)$.

\subsection{The dynamic algorithm}
\label{sub:sec:algo}


 Recall that $A, P$, and $D$ respectively denote the set of active, passive and deleted elements.    We let $A_i, P_i$ and $D_i$ respectively denote the set of active, passive  and dead elements at level $i \in [0, L]$. Let $\E_i = \{ e \in \E : \ell(e) = i\}$ denote the set of all elements in the current input that are at level $i \in [0, L]$.
Thus, for each $i \in [0, L]$, the set $\E_i$ is partitioned into two subsets: $A_i$ and $P_i$. For each level $i \in [0, L]$, we also define:
\begin{eqnarray}
\E_{\leq i} = \cup_{j \leq j} \E_j, \ \  A_{\leq i} = \cup_{j \leq i} A_j, \  \ P_{\leq i} = \cup_{j \leq i} P_j, 
 \text{ and } \ D_{\leq i} = \cup_{j \leq i} D_j. \label{eq:hyper-edges}
 \end{eqnarray}
For every level  $i \in [0, L]$, we maintain a counter $\C_{\leq i}$. Each call to   $\text{{\sc Rebuild}}(\leq i)$ sets $D_{\leq i} = P_{\leq i} = \emptyset$   and $\C_{\leq j} = \epsilon \cdot \left| \E_{\leq j} \right|$ for all $j \leq i$. In contrast,  every time an element  gets deleted from some level $i$, for all $j \in [i,L]$  we decrease the counter $\C_{\leq j}$  by one. Finally,  to ensure that the shadow input $(\S, \E^*)$ remains {\em close} to the actual input $(\S, \E)$,  we call {\sc Rebuild}($\leq i$)  whenever  $\C_{\leq i}$ becomes equal to $0$. If the counters of multiple levels become 0 during the same update, we call {\sc Rebuild}($\leq i$) for the largest such level $i$.

\smallskip
\noindent {\bf Tight sets:}
As in Section~\ref{sub:sec:shadow}, we will let $\S^*_{tight} = \{ s \in \S : (1+\epsilon)^{-1} c_s \leq W^*(s) \leq c_s \}$ denote the collection of tight sets with respect to the shadow input $(\S, \E^*)$.

\smallskip
\noindent {\bf Preprocessing phase:} Initially, we have an input $(\S, \E)$ where $\E = \emptyset$. At this point,  $D = \emptyset$, every set $s \in \S$ is at level $\ell(s) = 0$ with weight $W^*(s) = 0$, and hence Invariants~\ref{inv:element:weight:level:shadow},~\ref{inv:node:weight:shadow},~\ref{inv:hyper-edge:shadow} are vacuously true.

\smallskip
\noindent {\bf Handling the  deletion of an element:} When an element $e$ gets deleted, we call the subroutine described in Figure~\ref{fig:deletion}. Steps 01 -- 02 in Figure~\ref{fig:deletion} were  explained under {\em Case (a)} in Section~\ref{sub:sec:classification}, whereas  steps 03 -- 07 in Figure~\ref{fig:deletion} were explained while defining the counters $\C_{\leq i}$.

  \begin{figure*}[htbp]
                                                \centerline{\framebox{
                                                                \begin{minipage}{5.5in}
                                                                        \begin{tabbing}                                                                            
                                                                                01. \ \ \=  Remove the element $e$ from  $\E_{\ell(e)}$, and from $A_{\ell(e)} \cup P_{\ell(e)}$. \\
                                                                                02. \> Insert the element $e$ into  $D_{\ell(e)}$.   \\
                                                                                03. \>  {\sc For} $k = L$ down to $\ell(e)$: \\
                                                                                04. \> \ \ \ \  \ \  \= $\C_{\leq k} \leftarrow \C_{\leq k} - 1$. \\
                                                                                05. \> \>  {\sc If} $\C_{\leq k} = 0$, {\sc Then} \\
                                                                                06. \> \> \ \ \ \ \ \ \ \ \= Call the subroutine {\sc Rebuild}($\leq k$). \\ 
                                                                                07. \> \> \>  RETURN.                                                                 
                                                                        \end{tabbing}
                                                                \end{minipage}
                                                        }}
                                                        \caption{\label{fig:deletion} Handling the deletion of an element $e$.}
                                                \end{figure*}

\smallskip
\noindent {\bf Handling the  insertion of an element:}
When an element $e$ gets inserted, we call the subroutine in Figure~\ref{fig:insertion}. Steps 02 -- 04 and  05 -- 09 in Figure~\ref{fig:insertion} were respectively discussed  under {\em Case (b)} and {\em Case (c)} in Section~\ref{sub:sec:classification}.

 \begin{figure*}[htbp]
                                                \centerline{\framebox{
                                                                \begin{minipage}{5.5in}
                                                                        \begin{tabbing}
                                                                                01.  \ \ \= Let $i = \max\{ \ell(s) : s \in \S, e \in s \}$. \\
                                                                                02. \> {\sc If} there is at least one set $s \in \S \cap \S^*_{tight}$ that contains $e$, {\sc Then} \\                                                                                
                                                                                03. \> \ \ \ \ \ \ \ \= $\ell(e) \leftarrow i$. \\
                                                                                04. \> \> Insert the element $e$ into $\E_i$ and into $P_{i}$, with weight $w(e) \leftarrow 0$. \\                                                                                                                                                              
                                                                                05. \> {\sc Else} \\ 
                                                                                06. \> \> Let $\S_e = \ \{  s \in \S, e \in s\}$ be the collection of all sets that contain $e$. \\
                                                                                \> \>  Let $\lambda = \min \{ x : W^*(s) + x = c_s \text{ for some } s \in \S_e  \}$. \\
                                                                                07. \> \>  $\ell(e) \leftarrow i$.    \\
                                                                                08. \> \> Insert the element $e$ into $\E_i$ and into  $P_i$, with weight $w(e) \leftarrow  \lambda$. \\
                                                                                09. \> \> {\sc For} all sets $s \in \S_e$:  $W^*(s) \leftarrow W^*(s) + w(e)$.                                                                                                                                                                                                                                                                                                                            
                                                                        \end{tabbing}
                                                                \end{minipage}
                                                        }}
                                                        \caption{\label{fig:insertion} Handling the insertion of an element $e$.}
                                                \end{figure*}

\smallskip
\noindent {\bf Output of our algorithm:}  We  maintain the collection of tight sets $\S^*_{tight} = \{ s \in \S : (1+\epsilon)^{-1} c_s \leq W^*(s) < c_s\}$. We show in Section~\ref{sec:analysis} that  $\S^*_{tight}$ is a $(1+\epsilon)f$-approximate minimum set cover in $(\S, \E)$.

\smallskip
\noindent {\bf Correctness of the invariants:} Suppose that Invariants~\ref{inv:element:weight:level:shadow},~\ref{inv:node:weight:shadow},~\ref{inv:hyper-edge:shadow} hold just before the deletion of an element $e$. This is handled by the subroutine in Figure~\ref{fig:deletion}. It is easy to check that steps 01 -- 02 in Figure~\ref{fig:deletion} do {\em not} lead to a violation of any invariant. This is because the element $e$ gets moved from  $A \cup P$ to   $D$, but its weight $w(e)$ remains the same, and it  still contributes to the weights $W^*(s)$ of all the sets $s$ containing $e$. 

\smallskip
Similarly, suppose that Invariants~\ref{inv:element:weight:level:shadow},~\ref{inv:node:weight:shadow} and~\ref{inv:hyper-edge:shadow} hold just before the insertion of an element $e$. We handle this insertion by calling the subroutine in Figure~\ref{fig:insertion}. Consider two possible cases.

\smallskip
\noindent {\em Case (1): Steps 02 -- 04 get executed in Figure~\ref{fig:insertion}.} In this case, the element $e$ becomes passive with weight $w(e) = 0$, and it belongs to at least one tight set. Thus, the weight $W^*(s)$ of every set $s \in \S$ remains unchanged, and the three invariants continue to remain satisfied. 

\smallskip
\noindent {\em Case (2): Steps 05 -- 09 get executed in Figure~\ref{fig:insertion}.} In this case, all the sets $s \in \S$ containing $e$ have weights $W^*(s) < (1+\epsilon)^{-1}c_s$ and are at level $0$ (see Invariant~\ref{inv:node:weight:shadow}) at the time  $e$ gets inserted. Let $\S'_e = \arg\min_{s \in \S_e} \{ c_s - W^*(s)\}$. After we assign weight $w(e) \leftarrow \lambda$ to the element $e$, every set $s \in \S'_e$ gets  weight $W^*(s) = c_s$, and every other set $s \in \S_e \setminus \S'_e$ continues to have   weight $W^*(s) <  c_s$ (even though its weight increased). The weights of the  sets $s \in \S \setminus \S_e$ do not change. This ensures that Invariants~\ref{inv:node:weight:shadow} and~\ref{inv:hyper-edge:shadow} continue to hold. Finally, revisiting the proof of Claim~\ref{cl:residual:weight}, we infer that Invariant~\ref{inv:element:weight:level:shadow} also continues to hold, since  $e$ becomes passive with weight $w(e) = \lambda \leq 1 = (1+\epsilon)^0 = (1+\epsilon)^{-\ell(e)}$.

\smallskip
To summarize, we conclude that if the subroutine {\sc Rebuild}($\leq j$) has the property that a call to this subroutine never leads to a violation of the invariants, then the invariants continue to hold all the time. 

\medskip
\noindent {\bf Data structures:}  We  use the following data structures.
For each level $i \in [1, L]$, we maintain the sets $\E_i, A_i, P_i$ and $D_i$ as doubly linked lists.   Each entry in each of these lists also maintains a bidirectional pointer to the corresponding element. Using these pointers, we can determine the state of a given element (e.g., whether it is active, passive or dead) and insert/delete it in a given list in $O(1)$ time. 

\smallskip
\noindent For every element $e \in \E \cup D$, we maintain its level $\ell(e)$ and weight $w(e)$. For every set $s \in \S$, we also maintain its level $\ell(s)$ and weight  $W^*(s)$ with respect to the shadow input $(\S, \E^*)$.  Finally, for every level $i \in [0, L]$, we maintain the counter $\C_{\leq i}$.

\subsection{The {\sc Rebuild($\leq k$)} subroutine}
\label{sub:sec:rebuild} 

A detailed description of the subroutine  appears in Section~\ref{app:rebuild}. Here, we summarize a few key  properties of this subroutine that will be heavily used in the analysis of our algorithm. Property~\ref{prop:rebuild:invariant} ensures that Invariants~\ref{inv:element:weight:level:shadow},~\ref{inv:node:weight:shadow},~\ref{inv:hyper-edge:shadow} do not get violated. Property~\ref{prop:rebuild:runtime} specifies the time taken to implement a call to the subroutine, and how the counters $\{\C_{\leq i}\}$ get updated as a result of this call. Property~\ref{prop:rebuild:elements}, on the other hand, explains how the subroutine changes the states and levels of different elements in the hierarchical partition. 

\begin{property}
\label{prop:rebuild:invariant}
If Invariants~\ref{inv:element:weight:level:shadow},~\ref{inv:node:weight:shadow},~\ref{inv:hyper-edge:shadow} were satisfied just before a call to the subroutine {\sc Rebuild}($\leq k$) for any $k \in [0, L]$, then these invariants continue to remain satisfied at the end of that call.
\end{property}

\begin{property}
\label{prop:rebuild:runtime}
Consider any level $k \in [0, L]$. The time taken to implement a  call to  {\sc Rebuild($\leq k$)} is proportional to $f$ times the number of elements in $\E_{\leq k} \cup D_{\leq k}$ in the beginning of the call,  plus $O(\log (Cn)/\epsilon)$. Furthermore, at the end of this call, we have $\C_{\leq j} = \epsilon \cdot \left| \E_{\leq j} \right|$ for all levels $j \in [0, k]$.
\end{property}

\begin{property}
\label{prop:rebuild:elements}
Consider any level $k \in [0, L]$ and any call to the subroutine {\sc Rebuild($\leq k$)}.

\smallskip
\noindent {\bf (1)} The call to {\sc Rebuild($\leq k$)} {\em cleans up} all the dead elements at level $\leq k$. Specifically, this means the following. Consider any element $e$ that belongs to $D_{\leq k}$ just before the call to {\sc Rebuild}($\leq k$). Then that element $e$ does {\em not} appear in any of the sets $A, P$, or $D$ at the end of the call.

\smallskip
\noindent {\bf (2)} The call to {\sc Rebuild($\leq k$)} converts some of the passive elements at level $\leq k$ to passive elements at level $k+1$, and the remaining passive elements at level $\leq k$ get converted into active elements at level $\leq k+1$. Specifically, let $Z$ denote the set of elements in $P_{\leq k}$ just before the call to {\sc Rebuild}($\leq k$). Then during the call to {\sc Rebuild}($\leq k$), a subset $Z' \subseteq Z$ of these elements  gets added to $P_{k+1}$, and the remaining elements  $e \in Z' \setminus Z$ get added to $A_{\leq k+1}$. 

\smallskip
\noindent
{\bf (3)} The call to {\sc Rebuild($\leq k$)} moves up some of the active elements at level $\leq k$ to level $k+1$, and the remaining active elements at level $\leq k$ continue to be active at level $\leq k$. In other words, the elements in $A_{\leq k}$ never go out of the set $A_{\leq k+1}$ during the call to {\sc Rebuild}($\leq k$).

\smallskip
\noindent {\bf (4)}  The call to {\sc Rebuild($\leq k$)} does not touch the elements at level $\geq k+1$. In other words, for any $i \geq k+1$, if an element $e$ belonged to $A_i$, $P_i$  or $D_i$ just before the call to {\sc Rebuild($\leq k$)}, then it continues to belong to the same set $A_i, P_i$ or $D_i$ at the end of the call to {\sc Rebuild($\leq k$)}.

\end{property}

\begin{corollary}
\label{cor:prop:rebuild:elements}
At the end of any call to {\sc Rebuild}($\leq k$), we have $D_j = P_j =  \emptyset$  for all  $j \in [0, k]$.
\end{corollary}

\begin{proof}
Follows from parts (1), (2) of Property~\ref{prop:rebuild:elements}.
\end{proof}

\section{Analysis of our dynamic algorithm}
\label{sec:analysis}

We start by proving some simple properties of our algorithm that will be useful in the subsequent analysis.  These properties formalize the intuition that the fractional packing solution maintained by the algorithm does not change significantly in between two successive calls to  {\sc Rebuild}($\leq j$)  at any level $j \in [0, L]$. This happens because of three main reasons (see Figure~\ref{fig:deletion}). First, we set $\C_{\leq k} = \epsilon \cdot \left| \E_{\leq k} \right|$ for all $k \in [0, j]$ at the end of each call to {\sc Rebuild}($\leq j$). Second, we decrement the counter $\C_{\leq k}$ for all $k \in [j, L]$ each time some element gets deleted from level $j$. Third, we call  {\sc Rebuild}($\leq k$) whenever $\C_{\leq k}$ becomes equal to $0$.

\smallskip
\noindent {\bf Notation:} Throughout the rest of this section, we use the superscript $(t)$ to denote the status of some set/counter at time-step $t$. For instance, the symbol $D_{\leq j}^{(t)}$ will denote the set of dead elements at level $\leq j$ at time-step $t$, and the symbol $\C_{\leq j}^{(t)}$ will denote the value of the counter $\C_{\leq j}$ at time-step $t$.

\begin{lemma}
\label{lem:stability}
Fix any level $j \in [0, L]$ and consider any two time-steps $t' < t$ that satisfy the following properties: (1) A call was made to the subroutine {\sc Rebuild}($\leq k$) for some $k \in [j, L]$ just before time-step $t'$. (2) No call was made to {\sc Rebuild}($\leq k$) for any $k \in [j, L]$ during the time-interval $[t', t]$. Let $M_{\leq j}^{(t' \rightarrow t)}$ denote the set of elements that got deleted from level $\leq j$ during the time-interval $[t', t]$. Then we have:
$$\left| M_{\leq j}^{(t' \rightarrow t)} \right| +  \C_{\leq j}^{(t)}  = \epsilon \cdot \left| A_{\leq j}^{(t')} \right|.$$
\end{lemma}

\begin{proof}
As the subroutine {\sc Rebuild}($\leq k$) was called for some  $k \in [j, L]$ just before time-step $t'$,   Property~\ref{prop:rebuild:runtime} implies that $\C_{\leq j}^{(t')} = \epsilon \cdot \left|A_{\leq j}^{(t')}\right|$. Next, note that during the time-interval $[t', t]$, no call is made to the subroutine  {\sc Rebuild}($\leq k$) for any $k \in [j, L]$. Hence, during this time-interval, the counter $\C_{\leq j}$ gets decremented by one iff  an element gets deleted from level $\leq j$ (see Figure~\ref{fig:deletion}), and the set  $M_{\leq j}^{(t' \rightarrow t)}$ consists precisely of these elements. Thus, we infer that: $\left| M_{\leq j}^{(t' \rightarrow t)} \right| + \C_{\leq j}^{(t)} = \C_{\leq j}^{(t')} = \epsilon \cdot \left| A_{\leq j}^{(t')} \right|$. 
\end{proof}

\begin{corollary}
\label{cor:stability:extra}
Consider any level $j \in [1, L]$ and  time-steps $t' < t$ as defined in Lemma~\ref{lem:stability}. Then we have:
$$\left| M_{\leq j}^{(t' \rightarrow t)} \right| \leq  \epsilon \cdot \left| A_{\leq j}^{(t')} \right|.$$
\end{corollary}

\begin{proof}
Note that the subroutine {\sc Rebuild}($\leq j$) gets called whenever $\C_{\leq j} = 0$ (see Figure~\ref{fig:deletion}), and before finishing its execution the subroutine resets $\C_{\leq j} = \epsilon \cdot \left| \E_{\leq j} \right|$ (see Property~\ref{prop:rebuild:runtime}). Thus,  the counter $\C_{\leq j}$ always remains nonnegative, and in particular we have $\C_{\leq j}^{(t)} \geq 0$. The corollary now follows from Lemma~\ref{lem:stability}.
\end{proof}

\begin{corollary}
\label{cor:stability}
Consider any level $j \in [1, L]$ and  time-steps $t' < t$ as defined in Lemma~\ref{lem:stability}. Then we have:
$$\left| D_{\leq j}^{(t)} \right| \leq \epsilon \cdot \left| A_{\leq j}^{(t')} \right|.$$
\end{corollary}

\begin{proof}
Since the subroutine {\sc Rebuild}($\leq k$) was called for some $k \in [j, L]$ just before time-step $t'$, Corollary~\ref{cor:prop:rebuild:elements} implies that $D_{\leq j}^{(t')} = \emptyset$. We now track how the set $D_{\leq j}$ changes during the time-interval $(t', t)$. 

Whenever an element $e$ gets deleted from level $\leq j$ during this time-interval, the element $e$ gets added to both the sets $D_{\leq j}$ and $M_{\leq j}^{(t' \rightarrow t)}$. On the other hand, whenever the subroutine {\sc Rebuild}($\leq k$) gets called for some $k \in [1, j-1]$, all the dead elements at level $\leq k$ get removed from the hierarchical partition (see part (1) of Property~\ref{prop:rebuild:elements}). Since $k < j$, such a call to {\sc Rebuild}($\leq k$) can potentially remove some elements from the set $D_{\leq j}$, but no element from $M_{\leq j}^{(t' \rightarrow t)}$ gets removed due to the call.

Since no call is made to the subroutine {\sc Rebuild}($\leq k$) for any $k \in [j, L]$ during the time-interval $[t', t]$, and since $D_{\leq j} = \emptyset$ at time-step $t'$, the preceding discussion implies that $D_{\leq j}^{(t)} \subseteq M_{\leq j}^{(t' \rightarrow t)}$. Thus, from Corollary~\ref{cor:stability:extra}, we get $\left| D_{\leq j}^{(t)} \right|  \leq \left| M_{\leq j}^{(t' \rightarrow t)} \right| \leq \epsilon \cdot \left| A_{\leq j}^{(t')}\right|$. 
\end{proof}

\begin{lemma}
\label{lem:stability:active}
Consider any level $j \in [1, L]$ and  time-steps $t' < t$ as defined in Lemma~\ref{lem:stability}. Then we have:
$$\left|A_{\leq j}^{(t)}\right| \geq (1- \epsilon) \cdot \left| A_{\leq j}^{(t')} \right|.$$
\end{lemma}

\begin{proof}
According to part (3) of Property~\ref{prop:rebuild:elements}, a call to  {\sc Rebuild}($\leq k$) for some $k \in [1, j-1]$ can never decrease the size of the set $A_{\leq j}$. Since no call was made to {\sc Rebuild}($\leq k$) for any $k \in [j, L]$ during the time-interval $[t', t]$, we conclude that: 
 During the time-interval $[t', t]$, the set $A_{\leq j}$ can decrease in size only via deletion of elements from level $\leq j$. Moreover, the set $M_{\leq j}^{(t' \rightarrow t)}$   contains all these deleted elements. Thus, we infer that: $A_{\leq j}^{(t')} \setminus A_{\leq j}^{(t)} \subseteq M_{\leq j}^{(t' \rightarrow t)}$. Applying Corollary~\ref{cor:stability:extra}, we now get:
$\left| A_{\leq j}^{(t')} \setminus   A_{\leq j}^{(t)} \right| \leq \left| M_{\leq j}^{(t' \rightarrow t)} \right| \leq  \epsilon \cdot \left|A_{\leq j}^{(t')}\right|$. In words, at most an $\epsilon$ fraction of the elements get deleted from the set $A_{\leq j}$ during the time-interval $[t', t]$. Hence, it follows that $\left| A_{\leq j}^{(t)} \right| \geq (1-\epsilon) \cdot \left| A_{\leq j}^{(t')} \right|$.
\end{proof}

\begin{corollary}
\label{cor:number:dirty:nodes}
At any time-step $t$  and any level $j \in [0, L]$ we have
$\left| D_{\leq j}^{(t)} \right| \leq 2\epsilon \cdot \left| A_{\leq j}^{(t)} \right|$.
\end{corollary}

\begin{proof}
Fix any level $j \in [0, L]$ and  time-step $t$. We will show that the lemma holds for level $j$ at time-step $t$. Let $t' < t$ be the last time-step before $t$ with the following property:  a call was made to  {\sc Rebuild}($\leq k$) for some $k \in [j, L]$ just before time-step $t'$. Thus,  during the time-interval $[t', t]$ no call was made to {\sc Rebuild}($\leq k$) for any $k \in [j, L]$. Hence, Corollary~\ref{cor:stability} and Lemma~\ref{lem:stability:active} imply that:
\begin{eqnarray*}
\left|D_{\leq j}^{(t)}\right| & \leq & \epsilon \cdot \left|A_{\leq j}^{(t')}\right|  \\
& = & \left(\epsilon/(1-\epsilon)\right) \cdot (1-\epsilon) \left|A_{\leq j}^{(t')}\right| \\
& \leq & \left(\epsilon/(1-\epsilon)\right) \cdot \left|A_{\leq j}^{(t)}\right| \\
& \leq & 2 \epsilon \cdot \left|A_{\leq j}^{(t)}\right|.
\end{eqnarray*}
 The last inequality holds as long as $\epsilon \leq 1/2$. Thus, we infer that $|D_{\leq j}^{(t)}| \leq 2\epsilon \cdot |A_{\leq j}^{(t)}|$ at time-step $t$.
\end{proof}

\subsection{Bounding the update time of our dynamic algorithm}
\label{sec:update:time}

\newcommand{\I}{\mathcal{I}}

\begin{theorem}
\label{th:update:time}
Our dynamic algorithm has an amortized update time of $O(f \log (Cn)/\epsilon^2)$.
\end{theorem}

\begin{proof}
Recall that every element belongs to at most $f$ sets. Hence, ignoring the  potential call to  {\sc Rebuild}($\leq k$),  it takes $O(f + L) = O(f + \log_{(1+\epsilon)} (Cn)) = O(f + \log (Cn)/\epsilon)$ time to implement all the steps  in Figure~\ref{fig:deletion} and Figure~\ref{fig:insertion}. In other words, the update time of our dynamic algorithm is dominated by the time spent on the calls to {\sc Rebuild}$(\leq k)$. Henceforth, we focus on bounding the time spent on these calls.

Fix any  $k \in [0, L]$ and consider any  call to  {\sc Rebuild}($\leq k$) that is made just after some time-step (say) $t$. Let $t' < t$ be the last time-step before $t$ with the following property:  a call was made to  {\sc Rebuild}($\leq j$) for some $j \in [k, L]$ just before time-step $t'$. Since $\E_{\leq k} = A_{\leq k} \cup P_{\leq k}$, Corollary~\ref{cor:prop:rebuild:elements} states that:
\begin{equation}
\label{eq:verynew:1}
D_{\leq k}^{(t')} = P_{\leq k}^{(t')}  = \emptyset, \text{ and hence } \E_{\leq k}^{(t')} = A_{\leq k}^{(t')}.
\end{equation}
Let $\mathcal{I}_{\leq k}^{(t' \rightarrow t)}$ denote the set of elements  that get inserted into the set-system $(\S, \E)$ at some level $\leq k$ during the time-interval $[t', t]$. Furthermore, as in Lemma~\ref{lem:stability}, let $M_{\leq k}^{(t' \rightarrow t)}$ denote the set of elements that get deleted from some level $\leq k$ during the time-interval $[t', t]$. During the same time-interval, no call was made to  {\sc Rebuild}($\leq j$) for any $j \in [k, L]$. Moreover,  Property~\ref{prop:rebuild:elements} implies that a call to the subroutine {\sc Rebuild}($\leq j$) for some $j \in [1, k-1]$ does not change the set of elements in $\E_{\leq k}$. Thus, during the time-interval $[t', t]$, the only way the set $\E_{\leq k}$ can increase in size is via insertions of elements  at levels $\leq k$. It follows that:  $\E_{\leq k}^{(t)} \subseteq \E_{\leq k}^{(t')}  \cup \I_{\leq k}^{(t' \rightarrow t)}$. Since $\E_{\leq k}^{(t')} = A_{\leq k}^{(t')}$ according to~(\ref{eq:verynew:1}), we get:
\begin{equation}
\label{eq:verynew:2}
 \E_{\leq k}^{(t)} \subseteq A_{\leq k}^{(t')}  \cup \I_{\leq k}^{(t' \rightarrow t)}
\end{equation}
We will show next that $|D_{\leq k}^{(t)}| \leq  3 \epsilon |A_{\leq k}^{(t')}|$.
Note that $D_{\leq k}^{(t)}  \subseteq D_{\leq k}^{(t')} \cup  M_{\leq k}^{(t' \rightarrow t)}$.
From Corollary~\ref{cor:number:dirty:nodes} it follows that $|D_{\leq k}^{(t')}| \leq 2 \epsilon |A_{\leq k}^{(t')}|$.
Furthermore, since a call was made to the subroutine {\sc Rebuild}($\leq k$) just after time-step $t$, it must be the case that $\C_{\leq k}^{(t)} = 0$. Thus, from Lemma~\ref{lem:stability} we infer that
\begin{equation}
\label{eq:verynew:3}
\left| M_{\leq k}^{(t' \rightarrow t)} \right| = \epsilon \cdot \left| A_{\leq k}^{(t')}\right|
\end{equation}

It follows that:
\begin{equation}
\label{eq:verynew:3b}
|D_{\leq k}^{(t)}| \leq 3 \epsilon |A_{\leq k}^{(t')}|
\end{equation}
Let $T$ denote the total ``{\em cost}" (update time) we pay for calling the subroutine {\sc Rebuild}($\leq k$) at time-step $t$. From~(\ref{eq:verynew:2}),~(\ref{eq:verynew:3b}) and Property~\ref{prop:rebuild:runtime}, we get:
\begin{eqnarray}
T & = & O\left( f \cdot \left| \E_{\leq k}^{(t)} \right| + f \cdot \left| D_{\leq k}^{(t)} \right| \right) + O\left(\frac{\log (Cn)}{\epsilon}\right) \nonumber \\
& = & O\left( f \cdot \left| A_{\leq k}^{(t')} \right| + f \cdot \left| \I_{\leq k}^{(t' \rightarrow t)} \right| + \frac{\log (Cn)}{\epsilon}\right) \label{eq:verynew:4}
\end{eqnarray}
After each update, our dynamic algorithm (see Figures~\ref{fig:deletion} and~\ref{fig:insertion}) makes at most one call  to the {\sc Rebuild}$(\leq i)$ subroutine, over all $i \in [0, L]$. Hence, we can safely ignore the term $O(\log (Cn)/\epsilon)$ in $T$ above, as this term gets subsumed within our desired update time bound of $O(f \log (Cn)/\epsilon^2)$. We {\em split-up} the remaining chunk of $T$ into two parts: $T_1 = O\left(f \cdot \left| A_{\leq k}^{(t')} \right|  \right)$ and $T_2 = O\left(f \cdot \left| \I_{\leq k}^{(t' \rightarrow t)} \right|  \right)$. We {\em charge} the cost $T_1$ (resp. $T_2$) by distributing it evenly among  the elements that get deleted from (resp. inserted into) level $\leq k$ during the time-interval $[t', t]$. We now bound the total charge  accumulated by an element in this fashion. 

First, note that as per~(\ref{eq:verynew:3}),  $\epsilon \cdot \left| A_{\leq k}^{(t')}\right|$ elements get deleted from level $\leq k$ during the time-interval $[t', t]$. When we distribute the cost $T_1$ evenly among them, each of these elements accumulate a charge of $O(f/\epsilon)$. Now, consider any element $e$ that accumulates some charge in this fashion due to the call to {\sc Rebuild}($\leq k$) just after time-step $t$. By definition, this element $e$ gets deleted during the time-interval $[t', t]$. Accordingly, it is not possible for the same element $e$ to accumulate a similar charge from the same level $k$ at some future time-step $t'' > t$.\footnote{If $e$ is inserted and deleted again later, we consider this to be a different (instance of the) element.} To summarize, an element $e$ gets charged at most once from a given level in this manner.

Next, note that when we distribute the cost $T_2$ evenly among the elements in $\I_{\leq k}^{(t' \rightarrow t)}$, each such element accumulates a charge of $O(f)$. Consider any element $e$ that accumulates some charge from level $k$ just after some time-step $t$ in this manner. By definition, this element got inserted during the time-interval $[t', t]$, and thus it will never get charged due to a call to {\sc Rebuild}($\leq k$) at some future time-step $t'' > t$. To summarize, here again we derive that an element $e$ gets charged at most once from a given level in this fashion.

From the discussion in the preceding two paragraphs, we conclude that any element $e$ accumulates a charge of at most $O(f/\epsilon + f) = O(f/\epsilon)$ from each level. Thus, the total charge accumulated by any element is at most $O((f/\epsilon) L) = O((f/\epsilon) \log_{(1+\epsilon)} (Cn)) = O(f \log (Cn)/\epsilon^2)$. This means that the amortized update time of our dynamic algorithm is also $O(f \log (Cn)/\epsilon^2)$.
\end{proof}

\subsection{Bounding the approximation ratio}
\label{sec:approx}

 Our main result is summarized in the theorem below.

\begin{theorem}
\label{th:approx}
In the hierarchical partition maintained by our dynamic algorithm, the  tight sets $\S^*_{tight} = \{ s \in \S : (1+\epsilon)^{-1} c_s \leq W^*(s) \leq c_s\}$ form a $(1+5\epsilon)f$-approximate minimum set cover in $(\S, \E)$.
\end{theorem}

We now give a high-level overview of the proof of the above theorem. First, recall that  $\E^* = \E  \cup D$. Hence, the element-weights $\{w(e)\}$ define a valid fractional packing in the shadow-input $(\S, \E^*)$ as per Invariant~\ref{inv:node:weight:shadow}, and the  sets in $\S^*_{tight}$ form a valid set cover in $(\S, \E^*)$ as per Invariant~\ref{inv:hyper-edge:shadow}. As per Invariant~\ref{inv:hyper-edge:shadow}, every element $e \in \E \cup D$ is contained in at least one tight set. In other words, the fractional packing $\{w(e)\}$ is {\em approximately maximal}, in the sense that every element belongs to at least one set whose weight cannot be increased by more than $(1+\epsilon)$-factor. Armed with this observation, it is not too difficult to show that the total cost of the dual set cover (defined by the tight sets) is within a multiplicative factor  $(1+\epsilon)f$ of the total weight of the fractional packing $\{w(e)\}$ in $(\S, \E^*)$. This already implies that the collection of tight sets $\S^*_{tight}$ forms a $(1+\epsilon)f$-approximate minimum set cover in the shadow input (see Lemma~\ref{lem:compare:shadow}). The key challenge now is to show that the sets in $\S^*_{tight}$ also constitute an approximately minimum  set cover in the actual input $(\S, \E)$.

To address this challenge,  we exploit the fact that the number of elements in $D$ is relatively small compared to the number of  elements in $A$ (see Corollary~\ref{cor:number:dirty:nodes}). This implies that the total weight of the elements in $D$ is also small compared to the total weight of the elements in $A$ (see Lemma~\ref{lem:weight:dirty:nodes}). Hence, even if we delete all the elements in $D$ from the fractional packing $\{w(e)\}$, the objective value of the resulting solution will remain close to the objective value of the original fractional packing, which in turn was within a factor $(1+\epsilon)f$ of the total cost of the sets in $\S^*_{tight}$. 
So the total weight of the new fractional packing (after deleting the elements in $D$) will be very close to  $(1+\epsilon)f \cdot c\left(\S^*_{tight}\right)$, where $c\left(\S^*_{tight}\right)$ is the total cost of the sets in $\S^*_{tight}$ (see  Corollary~\ref{cor:compare:shadow}). Now, this also happens to be  a valid fractional packing in the actual input $(\S, \E)$, because we already had $W^*(s) \leq c_s$ for all sets $s \in \S$ and removing the elements in $D$ will {\em not} increase the weights of the sets any further. On the other hand, the sets in $\S^*_{tight}$ form a {\em valid} set cover in the actual input $(\S, \E)$ as well, as every $e \in \E$ belongs
to a set in $\S^*_{tight}$, see  Invariant~\ref{inv:hyper-edge:shadow}. In other words, we have identified a valid fractional packing and a valid set cover in $(\S, \E)$ whose objective values are  within a $(1+O(\epsilon))f$-factor of each other. So the corresponding set cover  must be an approximately minimum set cover in  $(\S, \E)$.

\begin{lemma}
\label{lem:weight:dirty:nodes}
We always have $\sum_{e \in D} w(e) \leq 2\epsilon \cdot \sum_{e \in A} w(e)$.
\end{lemma}

\begin{proof}
We first express the weight of an element $e$ as a sum of {\em increments}, where each increment corresponds to a specific level $k \geq \ell(e)$. To be more precise, we define:
\begin{equation*}
\Delta_k = \begin{cases}
(1+\epsilon)^{- k} & \text{ for } k = L; \\
(1+\epsilon)^{-k} - (1+\epsilon)^{-(k+1)} & \text{ for } 0 \leq k < L.
\end{cases}
\end{equation*}
From Invariant~\ref{inv:element:weight:level:shadow}, we  conclude that:
\begin{eqnarray}
\label{eq:100}
w(e) & = & (1+\epsilon)^{-\ell(e)} = \sum_{k = \ell(e)}^{L} \Delta_k \text{ for all } e \in A. \\
w(e) & \leq & (1+\epsilon)^{-\ell(e)} = \sum_{k = \ell(e)}^{L} \Delta_k \text{ for all } e \in D. \label{eq:111}
\end{eqnarray}
Now, we derive that:
\begin{eqnarray*}
\sum_{e \in D} w(e) & \leq & \sum_{e \in D} \sum_{k=\ell(e)}^{L} \Delta_k   \\
& = &  \sum_{k = 0}^{L} \sum_{e \in D \, : \, \ell(e) \leq k} \Delta_k  \\
& = &  \sum_{k=0}^{L} \Delta_k \cdot \left| D_{\leq k} \right|   \\
&  \leq & \sum_{k=0}^{L} \Delta_k \cdot 2\epsilon \left| A_{\leq k} \right|    \\
& = & 2\epsilon \cdot  \sum_{k=0}^{L} \sum_{e \in A \, : \, \ell(e) \leq k} \Delta_k \\
& = &  2\epsilon \cdot \sum_{e \in A} \sum_{k=\ell(e)}^{L} \Delta_k \\
& = &  2\epsilon \cdot \sum_{e \in A} w(e).  \nonumber
\end{eqnarray*}
In the above derivation, the first inequality follows from~(\ref{eq:111}). The second inequality follows from Corollary~\ref{cor:number:dirty:nodes}.  Finally, the last equality follows from~(\ref{eq:100}), This concludes the proof of the lemma.
\end{proof}

\begin{lemma}
\label{lem:compare:shadow}
We have $\sum_{e \in \E^*} w(e) \geq   \left((1+\epsilon)f\right)^{-1} \cdot c\left(\S^*_{tight}\right)$.
\end{lemma}

\begin{proof}
By Invariant~\ref{inv:hyper-edge:shadow}, every element $e \in \E  \cup D = \E^*$ belongs to at least one set in $\S^*_{tight}$.  We  sum over the weights of  these elements. Since $f$ is an upper bound on the maximum frequency of an element,   we get:
\begin{eqnarray*}
\sum_{e \in \E^*} w(e)   & \geq & f^{-1}  \sum_{s \in \S} W^*(s) \geq f^{-1}  \sum_{s \in \S^*_{tight}} W^*(s) \\
& \geq & f^{-1}  \sum_{s \in \S^*_{tight}} (1+\epsilon)^{-1}  c_s \\
& = &  ((1+\epsilon)f)^{-1} \cdot c\left(\S^*_{tight}\right). 
\end{eqnarray*}
\end{proof}

\begin{corollary}
\label{cor:compare:shadow}
We have $\sum_{e \in \E} w(e) \geq \left((1+\epsilon)(1+2\epsilon)f\right)^{-1}  \cdot c\left(\S^*_{tight}\right)$. 
\end{corollary}

\begin{proof}
Using Lemma~\ref{lem:weight:dirty:nodes}, we derive that:
\begin{eqnarray}
\left(1+(2 \epsilon)^{-1} \right) \cdot \sum_{e \in A} w(e) & \geq & \frac{1}{(2\epsilon)} \cdot \left( \sum_{e \in D} w(e) + \sum_{e \in A} w(e) \right) \nonumber
\end{eqnarray}
Multiplying both sides in the above inequality by $2\epsilon (1+2\epsilon)^{-1} = (1+(2\epsilon)^{-1})^{-1}$, we get:
\begin{equation}
\label{eq:new:200}
\sum_{e \in A} w(e) \geq (1+2\epsilon)^{-1} \cdot \sum_{e \in A \cup D} w(e).
\end{equation}
Now, adding the weights of the passive elements -- given by $\sum_{e \in P} w(e)$ -- on both sides of~(\ref{eq:new:200}), we get:
\begin{equation}
\label{eq:new:201}
\sum_{e \in A \cup P} w(e)  \geq (1+2\epsilon)^{-1} \cdot \sum_{e \in A \cup P \cup D} w(e).
\end{equation}
Since $\E = A \cup P$ and $\E^* = A \cup P \cup D$, the corollary follows from~(\ref{eq:new:201}) and Lemma~\ref{lem:compare:shadow}.
\end{proof}

\medskip
\noindent {\em Proof of Theorem~\ref{th:approx}.} By Invariant~\ref{inv:hyper-edge:shadow}, every element $e \in  \E$ belongs to at least one set  in $\S^*_{tight}$. In other words, the collection of sets $\S^*_{tight}$ forms a valid set cover in the input $(\S, \E)$. Next, from Invariant~\ref{inv:node:weight:shadow} we infer the following bound on the weight $W(s)$ of any set $s \in \S$:
$$W(s) \leq W^*(s) \leq c_s \text{ for all sets } s \in \S.$$
The first inequality holds  since we consider the weights of the elements  $e \in \E^* = \E \cup D$ while calculating $W^*(s)$, whereas we  only consider the weights of the elements $e \in \E$ while calculating $W(s)$. Thus,   the  element-weights $\{ w(e) \}_{e \in \E}$ form a valid {\em fractional packing} in $(\S, \E)$.  Finally, Corollary~\ref{cor:compare:shadow}  implies that the size of this fractional packing  is at least  $\alpha$ times of the total cost of the set cover $\S^*_{tight}$ in $(\S, \E)$, where $\alpha = ((1+\epsilon)(1+2\epsilon) f)^{-1} \geq ((1+5\epsilon)f)^{-1}$ when $0 < \epsilon < 1/2$. 
Theorem~\ref{th:approx} now follows  from Lemma~\ref{lm:duality}.

\section{Describing the {\sc Rebuild($\leq k$)} subroutine}
\label{app:rebuild} 

\newcommand{\D}{\mathcal{D}^*_{\leq k}}

\renewcommand{\P}{\mathcal{P}^*_{\leq k}}
\newcommand{\A}{\mathcal{A}^*_{\leq k}}

The subroutine works in 8 steps.   Throughout this section, we use the symbols $\E^*_{\leq k}, \A, \P$ and $\D$ respectively to denote the status of the sets $\E_{\leq k}, A_{\leq k}, P_{\leq k}$ and $D_{\leq k}$ just before the call to the subroutine.

\medskip
\noindent {\bf A note on the invariants:} While going through the description of the subroutine below, it will be helpful to remember that Invariant~\ref{inv:element:weight:level:shadow} will continue to hold all the time. In contrast, Invariants~\ref{inv:node:weight:shadow} and~\ref{inv:hyper-edge:shadow} will continue to hold only for those elements and sets that remain {\em unaffected} (do not change their levels/weights) during the call to REBUILD($\leq k$). These two invariants will be satisfied by the affected elements and sets only at the end of the subroutine. 

\medskip
\noindent {\bf Step 1:} Scan through all the elements in $\E^*_{\leq k} \cup \D$ and identify the collection of sets $\S' = \{ s \in \S : \ell(s) \leq k, s \cap (\E^*_{\leq k} \cup \D) \neq \emptyset\}$. A set belongs to $\S'$ iff it is at level $\leq k$ at this point in time and contains at least one element from $\E^*_{\leq k} \cup \D$. These are the sets whose levels and weights will be affected due to the call to {\sc Rebuild}($\leq k$). 

\medskip
\noindent {\bf Remark:} Consider any element $e \in \E^*_{\leq k} \cup \D$. These are the elements that get affected due to the call to {\sc Rebuild}($\leq k$).  Since the level of an element is defined to be the maximum level among all the sets it belongs to, we make the following important observation that will be used throughout the rest of this section.

\begin{observation}
\label{ob:imp:101}
Every set $s \in \S$ that contains some element $e \in \E^*_{\leq k} \cup \D$ belongs to the collection $\S'$. Furthermore, for every element $e' \in \E \cup D$ such that every set containing $e'$ belongs to $\S'$, we must have $e' \in \E^*_{\leq k} \cup \D$. 
\end{observation}

\medskip
\noindent {\bf Step 2:} Remove all the elements from $D_{\leq k}$, and accordingly modify the weights $W^*(s)$ of the  sets in $\S'$. Thus, we get $D_{\leq k} = D_{\leq k} \setminus \D = \emptyset$. (See part (1) of Property~\ref{prop:rebuild:elements}.)

\medskip
\noindent {\bf Step 3:} For every  element $e \in \P$, set $w(e) \leftarrow 0$  and accordingly modify the weight $W^*(s)$ of each set $s \in \S'$ that contains $e$. Now, move every set $s \in \S'$ to level $(k+1)$, by setting $\ell(e) \leftarrow (k+1)$ for all $s \in \S'$. Since the level of an element is defined to be the maximum level among all the sets it belongs to, this implies that all the elements $e \in \P \cup \A$ also move up to level $k+1$ (see Observation~\ref{ob:imp:101}). Each element $e \in \P$  becomes part of $P_{k+1}$, whereas each element $e \in \A$  becomes part of $A_{k+1}$.  Thus, at this point in time we get $P_{\leq k} = A_{\leq k} = D_{\leq k} = \emptyset$. We continue to have $w(e) = 0$ for all $e \in \P$ (this does not violate Invariant~\ref{inv:element:weight:level:shadow} since we now have $P^*_{\leq k} \subseteq P_{k+1}$). However,  to ensure that Invariant~\ref{inv:element:weight:level:shadow} holds for the active elements, we now set $w(e) = (1+\epsilon)^{-(k+1)}$ for all $e \in \A$, and we accordingly modify the weights $W^*(s)$ of the  sets in $\S'$.

\medskip
\noindent {\bf Remark:}  Steps 2 and 3  can only decrease the weights $W^*(s)$ of the sets $s \in \S'$. This is because just before step 2 we had $w(e) \geq (1+\epsilon)^{-k}$ for all elements $e \in \A$, and $w(e) \geq 0$ for all elements $e \in \P$. We also had $W^*(s) \leq c_s$ for all sets $s \in \S'$, as per Invariant~\ref{inv:node:weight:shadow}. Now, Step $2$ removed the elements from $D_{\leq k}$ and Step $3$ decreased the weights of the elements in $\A \cup \P$. Thus, we continue to have $W^*(s) \leq c_s$ for all sets $s \in \S'$ even after  Step 3. 

\medskip
\noindent {\bf Step 4:} {\sc For} every element $e \in \P$:
\begin{itemize}
\item Let $\S'_e \subseteq \S'$ denote the collection of all sets that contain $e$. Note that $\left| \S'_e \right| \leq f$.
\item (a) {\sc If} $W^*(s) \leq c_s - (1+\epsilon)^{-(k+1)}$ for all sets $s \in \S'_e$, {\sc Then}
\begin{itemize}
\item Increase the  weight of element $e$ from $0$ (see Step 3 above) to $w(e) \leftarrow (1+\epsilon)^{-(k+1)}$,  move the element $e$ from   $P_{k+1}$ to  $A_{k+1}$, and accordingly modify the weight $W^*(s)$ of every set $s \in \S'_e$. 
\end{itemize}
\item (b) {\sc Else}
\begin{itemize}
\item Let $\lambda = \min\{ x : W^*(s) + x = c_s \text{ for some } s \in \S'_e\}$.  Note that in this case  $\lambda < (1+\epsilon)^{-(k+1)}$.
\item Increase the weight of element $e$ from $0$ (see Step 3 above) to $w(e) \leftarrow \lambda$,
and accordingly modify the weight $W^*(s)$ of every set $s \in \S'_e$.
\end{itemize}
\end{itemize}

\medskip
\noindent {\bf Remark:} In Step 4 above, we set a weight $w(e) = (1+\epsilon)^{-(k+1)}$ to an element $e \in \P$ (thereby making it active) only if $W^*(s) \leq c_s - (1+\epsilon)^{-(k+1)}$ for all sets $s \in \S'$ containing $e$. Thus, even after Step $4$, we continue to have $W^*(s) \leq c_s$ for all sets $s \in \S'$. 

\medskip
\noindent {\bf Defining the collection of tight sets $\S'' \subseteq \S'$:} At this point in time, let $\S'' = \{ s \in \S' : (1+\epsilon)^{-1} c_s \leq  W^*(s) \leq c_s\}$  denote the collection of sets in $\S'$ whose weights lie between $(1+\epsilon)^{-1} c_s$ and $c_s$. The remark above implies that every remaining set $s \in \S' \setminus \S''$ has weight $W^*(s) < (1+\epsilon)^{-1} c_s$ at this point in time. We now prove two important claims.

\begin{claim}
\label{cl:new:201}
Just after the end of Step 4, each element $e \in \P \cap P_{k+1}$ belongs to at least one set $s \in \S''$ and has weight $w(e) \leq (1+\epsilon)^{-(k+1)}$ (so that it continues to satisfy Invariant~\ref{inv:element:weight:level:shadow}). 
\end{claim}

\begin{proof}
Consider any element $e \in \P \cap P_{k+1}$.  This means that the element $e$ was processed under case (b) in Step 4, because if it were processed under case (a) then it would no longer be part of $P_{k+1}$ at the end of Step 4.  
Accordingly, recall what we do with such an element $e$ under case (b) in Step 4. When we assign a weight $\lambda$ to the element $e$, at least one set $s' \in \S'_e$ gets a weight $W^*(s') = c_{s'}$, and that set $s'$ becomes part of $\S''$.

Next, note that if an element $e \in \P$ was processed under case (b) in Step 4, then it receives weight $w(e) \leq \lambda \leq (1+\epsilon)^{-(k+1)}$. Since such an element gets added to the set $P_{k+1}$ in Step 5, it continues to satisfy Invariant~\ref{inv:element:weight:level:shadow}.
\end{proof}

\begin{claim}
\label{cl:new:202}
Just after the end of Step 4, each element $e \in \P \setminus P_{k+1}$ belongs to $A_{k+1}$ and has weight $(1+\epsilon)^{-(k+1)}$ (so that it continues to satisfy Invariant~\ref{inv:element:weight:level:shadow}). 
\end{claim}

\begin{proof}
Consider any element $e \in \P \cap P_{k+1}$.  Such an element $e$ was processed under case (a) in Step 4. The claim follows from the description of that case.
\end{proof}

\medskip
\noindent {\bf Step 5:} Move each set $s \in \S' \setminus \S''$ down to level $k$, by setting $\ell(s) \leftarrow k$ for all $s \in \S' \setminus \S''$. Since the level of an element is defined to be the maximum level among all the sets it belongs to, some  elements  also move down to level $k$ along with the sets in $\S' \setminus \S''$. Let $X \subseteq E \cup D$ denote the subset of precisely those elements that move down to level $k$. Observation~\ref{ob:imp:101}, Claim~\ref{cl:new:201} and Claim~\ref{cl:new:202} imply that $X \subseteq \A \cup (\P \setminus P_{k+1}) = X \subseteq \A \cup (\P \cap A_{k+1})$. In other words, if an element $e$ moves down to level $k$ during this step, then it must be the case that: (a) $e$ is active right now, and (b) $e \in \A \cup \P$. Each element $e \in X$  becomes part of $A_k$. To ensure that Invariant~\ref{inv:element:weight:level:shadow} holds, we  set $w(e) \leftarrow (1+\epsilon)^{-k}$ for all  $e \in X$, and accordingly modify the weights $W^*(s)$ of the concerned sets in $\S' \setminus \S''$.

\medskip
\noindent {\bf Remark:} Consider a set $s \in \S' \setminus \S''$ that moved down to level $k$ in Step 5. It follows that at the end of Step 4, we had $W^*(s) < (1+\epsilon)^{-1}c_s$. During Step 5, we change the weights of some elements contained in $s$ from $(1+\epsilon)^{-(k+1)}$ to $(1+\epsilon)^{-k}$. This increases its weight $W^*(s)$ by at most a multiplicative factor of $(1+\epsilon)$. Thus, we continue to have $0 \leq W^*(s) \leq c_s$ at the end of Step 5 for all sets $s \in \S' \setminus \S''$.

\medskip
\noindent {\bf Summary of the situation after Step 5:}
To summarize, at the end of Step 5 we end up with the following situation. (a) Invariants~\ref{inv:element:weight:level:shadow},~\ref{inv:set:weight:level:shadow},~\ref{inv:element:level:shadow} hold for every element and set at level $\geq k+1$. (b) All the remaining, relevant sets $s \in \S' \setminus \S''$ are at level $k$ with weight $0 \leq W^*(s) \leq c_s$.  All the elements $e \in \E_{\leq k}$ are  at level $k$, with weight $w(e) = (1+\epsilon)^{-k}$. (c) There is no passive or dead element at level $\leq k$, that is, we have $P_{\leq k} = D_{\leq k} =  \emptyset$. (d) Finally, until this point the total time spent in the call to  {\sc Rebuild}($\leq k$) is proportional to $f$ times the number of elements in $\E_{\leq k} \cup D_{\leq k}$ in the beginning of the call (see Property~\ref{prop:rebuild:runtime}).  It now remains to fix the hierarchical partition at levels $\leq k$, by calling a subroutine that is very similar to the static algorithm from Appendix~\ref{app:sec:static}. 

\medskip
\noindent {\bf Step 6:} Call the subroutine FIX-LEVEL($k, \S'  \setminus \S'', \E_{\leq k}$). See Section~\ref{sec:efficient:preprocessing} for the details.


\medskip
\noindent {\bf Step 7:} For all levels $j \leq k$, set $\C_{\leq j} = \epsilon \cdot |\E_{\leq j}|$  (see Property~\ref{prop:rebuild:runtime}). 

\medskip
\noindent {\bf Step 8:} RETURN.

\subsection{Justifying Properties~\ref{prop:rebuild:invariant},~\ref{prop:rebuild:runtime} and~\ref{prop:rebuild:elements}}

\medskip
\noindent {\em Proof sketch for Property~\ref{prop:rebuild:invariant}:} Claim~\ref{cl:new:201}, Claim~\ref{cl:new:202} and the descriptions of Step 3 and Step 5 demonstrate that Invariant~\ref{inv:element:weight:level:shadow} holds at the end of Step 5. Further, it is easy to check that Invariant~\ref{inv:node:weight:shadow} and Invariant~\ref{inv:hyper-edge:shadow} hold for all the elements and sets at level $\geq k+1$ at the end of Step 5. Now, as explained in Section~\ref{sec:efficient:preprocessing}, the subroutine called in Step 6 simply gives a fast implementation of the static algorithm (see the Appendix) from level $k$ downward. This static algorithm moves all the slack sets  down to level $0$ (see Claim~\ref{app:prop:node:weight} in the Appendix).  This ensures that at the end of Step 6 all three invariants are satisfied.

\medskip
\noindent {\em Proof sketch for Property~\ref{prop:rebuild:runtime}:} The time taken to implement Steps 1 -- 5 is clearly proportional to $f$ times the size of the set $\E_{\leq k} \cup D_{\leq k}$ in the beginning of the call to REBUILD($\leq k$). The additive $O(\log (Cn)/\epsilon)$ term comes from the runtime analysis of the subroutine called in Step 6 (see Section~\ref{sec:efficient:preprocessing}). Finally, Step 7 ensures that $\C_{\leq j} = \epsilon \cdot |\E_{\leq j}|$ for all $j \in [0, k]$ at the end of the call to REBUILD($\leq k$).

\medskip
\noindent {\em Proof sketch for Property~\ref{prop:rebuild:elements}-(1):}  Follows from Step 2.

\medskip
\noindent {\em Proof sketch for Property~\ref{prop:rebuild:elements}-(2):} Claim~\ref{cl:new:201} and Claim~\ref{cl:new:202} imply the following: Consider any element $e$ that belonged to $P_{\leq k}$ just before the call to REBUILD($\leq k$). At the end of Step 4, either (a) the element $e$ is active or (b) the element $e$ is passive and it belongs to at least one set in $\S''$. Under case (a), it is easy to check that the element $e$ continues to remain active throughout the remainder of the call to REBUILD($\leq k$). Under case (b), Step 5 ensures that the element $e$ (along with the set in $\S''$ it belongs to) remains at level $(k+1)$ throughout the remainder of the call to REBUILD($\leq k$).

\medskip
\noindent {\em Proof sketch for Property~\ref{prop:rebuild:elements}-(3):} Note that every element $e \in \A$ continues to remain in $A_{\leq k+1}$ at the end of Step 5. This is because Steps 1 -- 5 do {\em not} ever change the state of an element from active to passive, Step 3 only moves some elements from level $\leq k$ to level $k+1$, and Step 5 again moves some of these elements back to level $k$. Finally, the subroutine in Step 6 never changes the state of an element from active to passive and never changes the level of any element at level $\geq k+1$ (see Section~\ref{sec:efficient:preprocessing}).

\medskip
\noindent {\em Proof sketch for Property~\ref{prop:rebuild:elements}-(4):} None of the Steps 1 -- 5 affects any element that was at level $\geq k+1$ just before the call to REBUILD($\leq k$). The same holds true for the subroutine in Step 6 (see Section~\ref{sec:efficient:preprocessing}).

\subsection{The subroutine FIX-LEVEL($k, \S', \E'$)}
\label{sec:efficient:preprocessing}

The most natural way to think about this subroutine is as follows. Suppose that we are executing the algorithm described in Appendix~\ref{app:sec:static}, and {\em so far we have have fixed everything above level $k$}. Let $\S'$ and $\E'$ be the remaining sets and elements whose levels are still undecided, meaning that every set $s \in \S \setminus \S'$ and every element $e \in \E \setminus \E'$ have already been assigned to some levels $> k$. {\bf Throughout the rest of this section, we assume that $|\S'| = m'$ and $|\E'| = n'$.} At the present moment, all the elements in $\E'$ have weight $(1+\epsilon)^{-k}$, because we know for sure that these elements will eventually get assigned to some level $\leq k$. Furthermore, at the present moment every set $s \in \S'$ has weight $W^*(s) \leq c_s$.\footnote{Note that in Appendix~\ref{app:sec:static} we denoted the weight of a set $s$ by $W(s)$, but here we are denoting its weight by $W^*(s)$. This is because here a part of the weight $W^*(s)$ might be coming from the weight of dead elements at level $> k$.} In order to ensure Property~\ref{prop:rebuild:runtime}, our task now is to construct the levels $\leq k$ of the hierarchical partition in $O(f n' + \log (Cn)/\epsilon)$ time.  If we follow the exact procedure in Appendix~\ref{app:sec:static}, then we will need $O(k \cdot f n')$ time for performing this task, because in that procedure we need to pay $O(f  n')$ time per level and we have to construct $k$ levels overall. Unfortunately, the resulting running time can  be as large as $\Omega(f n' \cdot \log (Cn)/\epsilon)$ because $k$ can be as large as $\Omega(L) = \Omega(\log (Cn)/\epsilon)$. The main challenge, therefore, is to come up with a procedure for our task that is much faster than the one described in Appendix~\ref{app:sec:static}.

Our new algorithm will maintain a partition of the collection of elements $\E'$  into two subsets: $\E'_{frozen} \subseteq \E'$  and $\E'_{alive} = \E' \setminus \E'_{frozen}$. Similarly, it will  maintain a partition of the collection of sets $\S'$ into two subsets: $\S'_{frozen} \subseteq \S'$ and $\S'_{alive} = \S' \setminus \S'_{frozen}$. In the very beginning, it will start by setting $\E' = \E'_{alive}$,  $\S' = \S'_{alive}$ and  $\E'_{frozen} = \S'_{frozen} = \emptyset$.  Intuitively, the levels of all the sets in $\S'$ and all the elements in $\E'$ are undecided in the beginning. Whenever the algorithm decides the final level of an alive set (resp. element), it will make the set (resp. element) frozen from that point onward.  

We next introduce the crucial notion of the {\em target level} $\ell_T(s)$ of a set $s \in \S'$. It is defined as follows.



If  $s \cap \E'_{alive} = \emptyset$, then  $\ell_T(s) = 0$. Otherwise, we define $\ell_T(s)$ to be the maximum level   $i \in [1, k]$ such that 
$$W^*(s) + \left((1+\epsilon)^{-i} - (1+\epsilon)^{-k}\right) \cdot \left| s \cap \E'_{alive} \right| \geq  \frac{c_s}{(1+\epsilon)}.$$

We now explain the intuition behind this definition. Fix any set $s \in \S'_{alive}$. Suppose that there is no other alive set  that shares a common alive element with $s$, i.e., $s \cap s' \cap \E'_{alive} = \emptyset$ for all $s' \in \S'_{alive} \setminus \{s\}$. If this is the case, then a moment's thought will reveal that the set $s$ will get assigned to the level $\ell_T(s)$   at the end of the algorithm in Appendix~\ref{app:sec:static}. This holds for the following reason: the set $s$ will keep decreasing its level and the alive elements in $s$ will keep increasing their weights in powers of $(1+\epsilon)$ until the weight of $s$ (given by $W^*(s)$ here) becomes larger than or equal to $(1+\epsilon)^{-1} c_s$. Also, note that if we move an alive element down from level $k$ to some level $i \leq k$, then its weight increases by $(1+\epsilon)^{-i} - (1+\epsilon)^{-k}$. Hence, by definition, $\ell_T(s)$ is the maximum level $i$ with the following property: If we move down the set $s$ (along with all the alive elements contained in $s$) to level $i$, then $W^*(s)$ becomes $\geq (1+\epsilon)^{-1} c_s$. So the set $s$ will get assigned to level $\ell_T(s)$ at the end of the algorithm in Appendix~\ref{app:sec:static}.

Unfortunately, the above argument does not hold if the set $s$ has some alive element $e$ in common with some other alive set $s'$. This is because in the algorithm described in Appendix~\ref{app:sec:static}, the set $s'$ can get assigned to a level $i' > \ell_T(s)$. This creates some problem in the argument above, because there we assumed that we can move  {\em all} the alive elements in $s$ down to level $\ell_T(s)$. But now, when the set $s'$ gets stuck at level $i' > \ell_T(s)$, it {\em enforces} that the element $e \in s'$ also gets stuck at level $i'$. In other words, we cannot move the element $e$ down all the way to level $\ell_T(s)$. So even after the set $s$ moves down to level $\ell_T(s)$, its weight $W^*(s)$ might still remain below the threshold $(1+\epsilon)^{-1} c_s$.

Nevertheless, we can still salvage the situation because of the following fact. Suppose that $s$ is an alive set with the maximum possible value of $\ell_T(s)$. Then the objection in the preceding paragraph does not apply to the set $s$.\footnote{Because in the paragraph we crucially relied on the fact that $i' > \ell_T(s)$, which will not be the case if $\ell_T(s) \geq \ell_T(s')$.} Accordingly, we pick the alive set $s$ with maximum possible value of $\ell_T(s)$, and then we move that set $s$ down to level $\ell_T(s)$, along with all the alive elements contained in $s$. Then we move the set $s$ from $\S'_{alive}$ to $\S'_{frozen}$ and also move all the ``relevant" elements $e \in s \cap \E'_{alive}$ from $\E'_{alive}$ to $\E'_{frozen}$. For every relevant element $e$, we next  visit all the alive sets $s''$ that contain $e$, and accordingly update their target levels $\ell_T(s'')$ in light of the facts that (a) the weight $W^*(s'')$ has increased by $(1+\epsilon)^{-i} - (1+\epsilon)^{-k}$ as the element $e$ has moved down from level $k$ to level $i$, and (b) the element $e$ is no longer alive.

Now comes a very crucial observation. Suppose that we pick an alive set $s$ that maximizes $\ell_T(s)$ and assign the set $s$ (and all the alive elements contained in $s$) to level $\ell_T(s)$, as described in the previous paragraph. {\em This can only decrease the value of $\ell_T(s'')$ for the other alive sets $s''$.} To see why this is the case, fix any other alive element $s'' \neq s'$ and suppose that we had $\ell_T(s'') = i''$ just before we decided to move the set $s$. At that point, we had $i'' \leq \ell_T(s)$ by definition of $s$. Hence, only the following situation can occur as the set $s$ moves down to level $\ell_T(s)$. Prior to this event, the set $s''$ felt that its weight $W^*(s'')$ will exceed the threshold $(1+\epsilon)^{-1}c_{s''}$ if it can move down to level $i''$ (along with all the alive elements in $s''$). But now, after the set $s$ settles at level $\ell_T(s)$, the set $s''$ realizes that it can no longer take some element $e'' \in s''$ (that was alive just before we moved $s$) all the way down to level $i''$, because $e''$ also belonged to $s$ and now it has become frozen at a higher level $\ell_T(s) > i''$ with a smaller weight $w(e'') < (1+\epsilon)^{-i''}$. Thus, at this point (after we have moved $s$), if the set $s''$ decides to go down to level $i''$, then its weight $W^*(s)$ can only be less than what  it would have been prior to the instant we moved $s$. In other words, while moving $s$ down to level $\ell_T(s)$, we can only decrease the value of $\ell_T(s'')$, for the set $s''$ will now need to go down to an even lower level in order to ensure that $W^*(s)$ exceeds the threshold $(1+\epsilon)^{-1}c_{s''}$.

This leads us to the following natural algorithm. {\bf  Consider an array  $\Gamma[0, \ldots, k]$, where $\Gamma[i] = \{ s \in \S'_{alive} : \ell_T(s) = i\}$ for each $i \in [0, k]$.} As far as data structures are concerned, we can store each entry $\Gamma[i]$ of this array as a doubly linked list. Using appropriate pointers, we can insert/delete a given set $s'$ in such a doubly linked list in $O(1)$ time. Now, the algorithm proceeds in rounds $i = k, \ldots, 0$. In the beginning of round $i$, we have $\Gamma[j] = \emptyset$ for all $j > i$. During round $i$, we repeatedly keep pulling out a set $s$ from $\Gamma[i]$ (this is a set which currently maximizes $\ell_T(s)$),  move the set $s$ -- along with all the elements in $\E'_{alive} \cap s$ -- down to level  $\ell_T(s)$, transfer the set $s$ from $\S'_{alive}$ to $\S'_{frozen}$, and also transfer all the elements $e \in s \cap \E'_{alive}$ from $\E'_{alive}$ to $\E'_{frozen}$. By the discussion above, these steps can only decrease the values of $\ell_T(s'')$ of the other alive sets $s''$. Accordingly, we continue to satisfy the invariant that $\Gamma[j] = \emptyset$ for all $j > i$. The current round ends when we have $\Gamma[i] = \emptyset$. At that point, we proceed with round $(i-1)$.

From the above discussion, it becomes clear that this algorithm produces exactly the same output as the algorithm described in Appendix~\ref{app:sec:static}. It now remains to bound the total runtime of this algorithm. The pseudocode of the algorithm appears in Figure~\ref{fig:new:preprocessing}.


   \begin{figure*}[htbp]
                        \centerline{\framebox{
                        \begin{minipage}{5.5in}
                        \begin{tabbing}                                                                                
                                                  01.  \ \=  {\sc Initialize:} $\S'_{alive} \leftarrow \S'$, $\S'_{frozen} \leftarrow \emptyset$, $\E'_{alive} \leftarrow \E'$, $\E'_{frozen} \leftarrow \emptyset$. \\
                                                  02. \>  Compute $\ell_T(s)$ for every set $s \in \S'_{alive}$, and  set up the array $\Gamma[0, \ldots, k]$ accordingly. \\
                                                  05. \> {\sc For} $i = k$ down to $0$: \\      
                                                  06. \> \ \ \ \ \ \= {\sc While} $\Gamma[i] \neq \emptyset$:  \\
                                                  07. \> \> \ \ \ \ \ \= Pick any  $s \in \Gamma[i]$. \\
                                                  08. \> \> \> $\Gamma[i] \leftarrow \Gamma[i] \setminus \{s\}$.         \\
                                                  09. \> \> \> $\ell(s) \leftarrow \ell_T(s)$. \\
                                                  10. \> \> \> Move the set $s$ from $\S'_{alive}$ to $\S'_{frozen}$. \\
                                                  11. \> \> \> {\sc For} every element $e \in \E'_{alive} \cap s$: \\
                                                  12. \> \> \> \ \ \ \  \=  $\ell(e) \leftarrow \ell_T(s)$. \\
                                                  13. \> \> \> \> $w(e) \leftarrow w(e) + (1+\epsilon)^{-\ell_T(s)} - (1+\epsilon)^{-k}$.   \ \ \  // We had $w(e) = (1+\epsilon)^{-k}$ before this step. \\
                                                  14.\> \> \> \> $W^*(s) \leftarrow W^*(s) + (1+\epsilon)^{-\ell_T(s)} - (1+\epsilon)^{-k}$. \\
                                                  15. \> \> \> \> Move the element $e$ from $\E'_{alive}$ to $\E'_{frozen}$. \\
                                                  16. \> \> \> \> {\sc For} every set $s' \in \S'_{alive}$ that contains $e$:  \\
                                                  17. \> \> \> \> \ \  \ \ \ \ \ \=$W^*(s') \leftarrow W^*(s')  + (1+\epsilon)^{-\ell_T(s)} - (1+\epsilon)^{-k}$. \\
                                                  18. \> \> \> \> \> Update the value of $\ell_T(s')$ and the array $\Gamma[0, \ldots, k]$ accordingly.                                            
                         \end{tabbing}
                         \end{minipage}
                         }}
                         \caption{\label{fig:new:preprocessing} The subroutine {\sc FIX-LEVEL}$(k, \S', \E')$.}
            \end{figure*}

\medskip
\noindent {\bf Runtime analysis:} Since each element is contained in at most $f$ sets, steps 12 -- 18 in Figure~\ref{fig:new:preprocessing} can be implemented in $O(f)$ time (note that step 18 can be implemented in $O(1)$ time). Furthermore, while executing steps 12 -- 18 on a given element $e$, we move the element from $\E'_{alive}$ to $\E'_{frozen}$. Hence, steps 12 -- 18 get executed on a given element $e$ at most once. Summing over all the $n'$ elements in $\E'$, the subroutine in Figure~\ref{fig:new:preprocessing} spends at most $O(f \cdot n')$ time on steps 11 -- 18.

Next, note that while executing steps 07 -- 10 on a given set $s$, the subroutine in Figure~\ref{fig:new:preprocessing} moves the set from $\S'_{alive}$ to $\S'_{frozen}$. Accordingly, steps 07 -- 10 get executed on a given set $s$ at most once. Since every such execution of steps 07 -- 10 takes $O(1)$ time and since there are $m'$ sets in $\S'$, we conclude that the subroutine overall spends $O(m')$ time on steps 06 -- 10.

Finally, the subroutine spends $O(k) \leq O(L) = O(\log (Cn)/\epsilon)$ time on step 05. Since $m' \leq f n'$, the total time required by the subroutine in Figure~\ref{fig:new:preprocessing} is at most $O(f n' + m' + \log (Cn)/\epsilon) = O(f n' + \log (Cn)/\epsilon)$.

\section*{Acknowledgement}  
We thank Xiaowei Wu for spotting an error in an earlier version of the paper. 

The research leading to these 
results has received funding from the European Research Council under 
the European Union's Seventh Framework Programme (FP/2007-2013) / ERC 
Grant Agreement no. 340506.   

This project has received funding from the European Research Council (ERC) under the European Union's Horizon 2020 research and innovation programme under grant agreement No 715672. Nanongkai was also supported by the Swedish Research Council (Reg. No. 2015-04659).


\bibliographystyle{plain}
\bibliography{IEEEabrv,bibliography}

\begin{thebibliography}{10}

\bibitem{AbboudAGPS-stoc19}
Amir Abboud, Raghavendra Addanki, Fabrizio Grandoni, Debmalya Panigrahi, and
  Barna Saha.
\newblock Dynamic set cover: Improved algorithms \& lower bounds.
\newblock In {\em {STOC}}. {ACM}, 2019.

\bibitem{AbboudW14}
Amir Abboud and Virginia~Vassilevska Williams.
\newblock Popular conjectures imply strong lower bounds for dynamic problems.
\newblock In {\em {FOCS}}, pages 434--443. {IEEE} Computer Society, 2014.

\bibitem{BaswanaGS18}
Surender Baswana, Manoj Gupta, and Sandeep Sen.
\newblock Fully dynamic maximal matching in o(log n) update time (corrected
  version).
\newblock {\em {SIAM} J. Comput.}, 47(3):617--650, 2018.
\newblock announced at FOCS'11.

\bibitem{BernsteinC16}
Aaron Bernstein and Shiri Chechik.
\newblock Deterministic decremental single source shortest paths: beyond the
  o(mn) bound.
\newblock In {\em {STOC}}, pages 389--397. {ACM}, 2016.

\bibitem{BernsteinC-soda17}
Aaron Bernstein and Shiri Chechik.
\newblock Deterministic partially dynamic single source shortest paths for
  sparse graphs.
\newblock In {\em {SODA}}, pages 453--469. {SIAM}, 2017.

\bibitem{BernsteinFH19}
Aaron Bernstein, Sebastian Forster, and Monika Henzinger.
\newblock A deamortization approach for dynamic spanner and dynamic maximal
  matching.
\newblock In {\em {SODA}}, pages 1899--1918. {SIAM}, 2019.

\bibitem{BernsteinS15}
Aaron Bernstein and Cliff Stein.
\newblock Fully dynamic matching in bipartite graphs.
\newblock In {\em {ICALP} {(1)}}, volume 9134 of {\em Lecture Notes in Computer
  Science}, pages 167--179. Springer, 2015.

\bibitem{BernsteinS16}
Aaron Bernstein and Cliff Stein.
\newblock Faster fully dynamic matchings with small approximation ratios.
\newblock In {\em {SODA}}, pages 692--711. {SIAM}, 2016.

\bibitem{BhattacharyaCH17}
Sayan Bhattacharya, Deeparnab Chakrabarty, and Monika Henzinger.
\newblock Deterministic fully dynamic approximate vertex cover and fractional
  matching in {O(1)} amortized update time.
\newblock In {\em {IPCO}}, volume 10328 of {\em Lecture Notes in Computer
  Science}, pages 86--98. Springer, 2017.

\bibitem{BhattacharyaHI18-SICOMP-Matching}
Sayan Bhattacharya, Monika Henzinger, and Giuseppe~F. Italiano.
\newblock Deterministic fully dynamic data structures for vertex cover and
  matching.
\newblock {\em {SIAM} J. Comput.}, 47(3):859--887, 2018.
\newblock announced at SODA'15.

\bibitem{BhattacharyaHI18}
Sayan Bhattacharya, Monika Henzinger, and Giuseppe~F. Italiano.
\newblock Dynamic algorithms via the primal-dual method.
\newblock {\em Inf. Comput.}, 261(Part):219--239, 2018.
\newblock announced at ICALP'15.

\bibitem{BhattacharyaHN16}
Sayan Bhattacharya, Monika Henzinger, and Danupon Nanongkai.
\newblock New deterministic approximation algorithms for fully dynamic
  matching.
\newblock In {\em {STOC}}, pages 398--411. {ACM}, 2016.

\bibitem{BhattacharyaHN17}
Sayan Bhattacharya, Monika Henzinger, and Danupon Nanongkai.
\newblock Fully dynamic approximate maximum matching and minimum vertex cover
  in \emph{O}(log\({}^{\mbox{3}}\) \emph{n}) worst case update time.
\newblock In {\em {SODA}}, pages 470--489. {SIAM}, 2017.

\bibitem{ChuzhoyK19}
Julia Chuzhoy and Sanjeev Khanna.
\newblock A new algorithm for decremental single-source shortest paths with
  applications to vertex-capacitated flow and cut problems.
\newblock In {\em {STOC}}. {ACM}, 2019.

\bibitem{Chvatal79}
Vasek Chv{\'{a}}tal.
\newblock A greedy heuristic for the set-covering problem.
\newblock {\em Math. Oper. Res.}, 4(3):233--235, 1979.

\bibitem{DinurGKR05}
Irit Dinur, Venkatesan Guruswami, Subhash Khot, and Oded Regev.
\newblock A new multilayered {PCP} and the hardness of hypergraph vertex cover.
\newblock {\em {SIAM} J. Comput.}, 34(5):1129--1146, 2005.
\newblock announced at STOC'03.

\bibitem{DinurS14}
Irit Dinur and David Steurer.
\newblock Analytical approach to parallel repetition.
\newblock In {\em {STOC}}, pages 624--633. {ACM}, 2014.

\bibitem{GLSSS19}
Fabrizio Grandoni, Stefano Leonardi, Piotr Sankowski, Chris Schwiegelshohn, and
  Shay Solomon.
\newblock {(1} + {\(\epsilon\)})-approximate incremental matching in constant
  deterministic amortized time.
\newblock In {\em {SODA}}, pages 1886--1898. {SIAM}, 2019.

\bibitem{GuptaK0P17}
Anupam Gupta, Ravishankar Krishnaswamy, Amit Kumar, and Debmalya Panigrahi.
\newblock Online and dynamic algorithms for set cover.
\newblock In {\em {STOC}}, pages 537--550. {ACM}, 2017.

\bibitem{GuptaP13}
Manoj Gupta and Richard Peng.
\newblock Fully dynamic {(1+} e)-approximate matchings.
\newblock In {\em {FOCS}}, pages 548--557. {IEEE} Computer Society, 2013.

\bibitem{HenzingerKN-FOCS14}
Monika Henzinger, Sebastian Krinninger, and Danupon Nanongkai.
\newblock Decremental single-source shortest paths on undirected graphs in
  near-linear total update time.
\newblock In {\em {FOCS}}, pages 146--155. {IEEE} Computer Society, 2014.

\bibitem{HenzingerKN-STOC14}
Monika Henzinger, Sebastian Krinninger, and Danupon Nanongkai.
\newblock Sublinear-time decremental algorithms for single-source reachability
  and shortest paths on directed graphs.
\newblock In {\em {STOC}}, pages 674--683. {ACM}, 2014.

\bibitem{HenzingerKN16}
Monika Henzinger, Sebastian Krinninger, and Danupon Nanongkai.
\newblock Dynamic approximate all-pairs shortest paths: Breaking the o(mn)
  barrier and derandomization.
\newblock {\em {SIAM} J. Comput.}, 45(3):947--1006, 2016.
\newblock announced at FOCS'13.

\bibitem{HenzingerKN18}
Monika Henzinger, Sebastian Krinninger, and Danupon Nanongkai.
\newblock Decremental single-source shortest paths on undirected graphs in
  near-linear total update time.
\newblock {\em J. {ACM}}, 65(6):36:1--36:40, 2018.

\bibitem{HenzingerKNS15}
Monika Henzinger, Sebastian Krinninger, Danupon Nanongkai, and Thatchaphol
  Saranurak.
\newblock Unifying and strengthening hardness for dynamic problems via the
  online matrix-vector multiplication conjecture.
\newblock In {\em {STOC}}, pages 21--30. {ACM}, 2015.

\bibitem{KhotR08}
Subhash Khot and Oded Regev.
\newblock Vertex cover might be hard to approximate to within 2-epsilon.
\newblock {\em J. Comput. Syst. Sci.}, 74(3):335--349, 2008.
\newblock announced at CCC'03.

\bibitem{NanongkaiS17}
Danupon Nanongkai and Thatchaphol Saranurak.
\newblock Dynamic spanning forest with worst-case update time: adaptive, las
  vegas, and o(n\({}^{\mbox{1/2 - {\(\epsilon\)}}}\))-time.
\newblock In {\em {STOC}}, pages 1122--1129. {ACM}, 2017.

\bibitem{NanongkaiSW17}
Danupon Nanongkai, Thatchaphol Saranurak, and Christian Wulff{-}Nilsen.
\newblock Dynamic minimum spanning forest with subpolynomial worst-case update
  time.
\newblock In {\em {FOCS}}, pages 950--961. {IEEE} Computer Society, 2017.

\bibitem{NanongkaiSY-stoc19}
Danupon Nanongkai, Thatchaphol Saranurak, and Sorrachai Yingchareonthawornchai.
\newblock Breaking quadratic time for small vertex connectivity and an
  approximation scheme.
\newblock In {\em {STOC}}. {ACM}, 2019.

\bibitem{NeimanS16}
Ofer Neiman and Shay Solomon.
\newblock Simple deterministic algorithms for fully dynamic maximal matching.
\newblock {\em {ACM} Trans. Algorithms}, 12(1):7:1--7:15, 2016.
\newblock Announced at STOC'13.

\bibitem{OnakR10}
Krzysztof Onak and Ronitt Rubinfeld.
\newblock Maintaining a large matching and a small vertex cover.
\newblock In {\em {STOC}}, pages 457--464. {ACM}, 2010.

\bibitem{PelegS16}
David Peleg and Shay Solomon.
\newblock Dynamic (1 + $\epsilon$)-approximate matchings: {A} density-sensitive
  approach.
\newblock In {\em {SODA}}, pages 712--729. {SIAM}, 2016.

\bibitem{Sankowski07}
Piotr Sankowski.
\newblock Faster dynamic matchings and vertex connectivity.
\newblock In {\em {SODA}}, pages 118--126. {SIAM}, 2007.

\bibitem{Slavik97}
Peter Slav{\'{\i}}k.
\newblock A tight analysis of the greedy algorithm for set cover.
\newblock {\em J. Algorithms}, 25(2):237--254, 1997.

\bibitem{Solomon16}
Shay Solomon.
\newblock Fully dynamic maximal matching in constant update time.
\newblock In {\em {IEEE} 57th Annual Symposium on Foundations of Computer
  Science, {FOCS}}, pages 325--334, 2016.

\bibitem{BrandNS-FOCS19}
Jan van~den Brand, Danupon Nanongkai, and Thatchaphol Saranurak.
\newblock Dynamic matrix inverse: Improved algorithms and matching conditional
  lower bounds.
\newblock In {\em {FOCS}}. {IEEE} Computer Society, 2019.

\end{thebibliography}

\appendix

\section{A Static Primal-Dual Algorithm for Minimum Set Cover}
\label{app:sec:static}

\begin{figure*}[h]
	\centerline{\framebox{
			\begin{minipage}{5.5in}
				\begin{tabbing}                                                                                
					1. \  \= {\sc Initialize}: \= $L = \lceil \log_{(1+\epsilon)} (C \cdot n) \rceil + 1$, \\
					\> \>  $w(e) \leftarrow (1+\epsilon)^{-L}$ for all elements $e \in \E$. \\
					\> \>  $\ell(s) \leftarrow L$ for every set $s \in \S$.     \\
					\> \> $\ell(e) \leftarrow L$ for every element $e \in \E$.   \\
					2. \> {\sc For} rounds $t = L$ to $1$:   \\
					3. \> \ \ \  \ \ \ \= Let $\S^{(t)}_{slack} = \{ s \in \S : W(s) < (1+\epsilon)^{-1} c_s\}$ be the collection of sets that are {\em slack}    \\
					\> \> in the   beginning of round $t$, and let $\E^{(t)}_{slack} = \{ e \in \E : e \notin s \text{ for all } s \in \S \setminus \S^{(t)}_{slack}\}$ be the \\
					\> \>  collection of elements that are {\em exclusively} covered by the sets in $\S^{(t)}_{slack}$.    \\
					4. \> \>  {\sc For all} sets $s \in \S^{(t)}_{slack}$:          \\
					5. \> \> \ \ \ \ \ \  \ \ \ \= $\ell(s) = \ell(s) - 1$.      \\
					6. \> \> {\sc For all} elements $e \in \E^{(t)}_{slack}$: \\
					7. \> \> \> $w(e) = (1+\epsilon) \cdot w(e)$.    \\
					8. \> \> \> $\ell(e) = \ell(e) -1$. \qquad \qquad // This ensures that $\ell(e) = \max\{ \ell(s) : s \in \S \text{ and } e \in s\}$.   
				\end{tabbing}
			\end{minipage}
	}}
	\caption{\label{app:fig:static} A static $(1+\epsilon)f$-approximation algorithm for minimum set cover.} 
\end{figure*}

We use the same notations as in Section~\ref{sec:preprocessing}.  The discretized primal-dual algorithm proceeds in $L$ rounds (see Figure~\ref{app:fig:static}). In the beginning, we start by assigning a weight $w(e) = (1+\epsilon)^{-L} \leq (Cn)^{-1}$ to every element $e \in \E$, and a {\em level} $\ell(s) = \ell(e) = L$ to every set $s \in \S$ and every element $e \in \E$. Since each set in $\S$ contains at most $n$ elements, at this point in time we have $W(s) = \sum_{e \in s} w(e) \leq n \cdot (Cn)^{-1} = 1/C$ for all sets $s \in \S$. Hence, from~(\ref{eq:weight}) we infer that $0 \leq W(s) \leq c_s$ for all sets $s\in \S$.  In other words, we have a valid fractional packing at this point in time. Throughout the rest of this section, we say that a set $s \in \S$ is {\em tight} if $(1+\epsilon)^{-1} c_s \leq W(s) \leq c_s$, and {\em slack} if $0 \leq W(s) < (1+\epsilon)^{-1}c_s$. 

In the beginning,  we start with a counter $t = L$. The value of $t$  keeps decreasing by one until we reach $t = 0$. Each value of $t$ corresponds to a distinct {\em round} in the algorithm.   In each round $t$, we increase (by a $(1+\epsilon)$ factor)  the weights of the elements that are {\em exclusively} covered by the slack sets, and we decrease by one the levels of the concerned sets (that were slack in the beginning of the current round) and  elements (whose weights got increased in the current round). Thus, there are $L$ rounds overall, one for each  $t \in \{L, \ldots, 1\}$. At the end of round $1$, we  have  a {\em hierarchical partition} of  $\S$ into $L+1$ {\em levels} $\{L, \ldots, 0\}$.  We now describe  a few key properties that are satisfied at the end of the  algorithm in Figure~\ref{app:fig:static}.

\begin{claim}
	\label{app:prop:hyper-edge:weight}
	For every element $e \in \E$, we have  $\ell(e) = \max \{ \ell(s) : s \in \S, e \in s\}$ and $w(e) = (1+\epsilon)^{-\ell(e)}$.
\end{claim}

\begin{proof}(Sketch)
	Fix any element $e \in \E$. Initially,  every set and every element is assigned to level $L$. See step 1 in Figure~\ref{app:fig:static}. At this point, we clearly have $\ell(e) = \max\{ \ell(s) : s \in \S, e \in s\}$. Subsequently, during a given round $t$, we decrement the level of $e$ only if $e \in \E_{slack}^{(t)}$. See step 8 in Figure~\ref{app:fig:static}. Now, note that if  $e \in \E_{slack}^{(t)}$, then by definition every set  $s \in \S$ containing $e$ belongs to $\S_{slack}^{(t)}$. Furthermore, as is evident from step 5 in Figure~\ref{app:fig:static}, during round $t$ we also decrement the level of every set $s \in \S^{(t)}_{slack}$. Thus, we always have: $\ell(e) = \max\{ \ell(s) : s \in \S, e \in s\}$. 
	
	Next, note that as per  step 1 in Figure~\ref{app:fig:static}, we initially have $\ell(e) = L$ and $w(e) = (1+\epsilon)^{-L} = (1+\epsilon)^{-\ell(e)}$. Subsequently, whenever we increase  $w(e)$ by a factor of $(1+\epsilon)$, we also decrement the level $\ell(e)$ by one (see steps 7, 8 in Figure~\ref{app:fig:static}). Hence, it follows that $w(e) = (1+\epsilon)^{-\ell(e)}$ at the end of the algorithm.
\end{proof}

\begin{claim}
	\label{app:prop:hyper-edge}
	For every element $e \in \E$, at least one of the sets containing $e$ lies at level $\geq 1$ (i.e.,  $\ell(e) \geq 1$).
\end{claim}

\begin{proof}(Sketch)
	Fix any element $e \in \E$. For the sake of contradiction, suppose that $\ell(e) = 0$ at the end of the algorithm. This implies that $e \in \E^{(1)}_{slack}$ in the beginning of round $1$. This is because only the elements in $\E^{(1)}_{slack}$ get moved down to level $0$. Hence, as per the proof of Claim~\ref{app:prop:hyper-edge:weight}, we have $w(e) = (1+\epsilon)^{-1}$ in the beginning of round $1$. This means that every set $s \in \S$ containing the element $e$ has  $W(s) \geq w(e) \geq (1+\epsilon)^{-1} \geq (1+\epsilon)^{-1} c_s$ in the beginning of round $1$. The last inequality holds since  $c_s \leq 1$ for all sets $s \in \S$, as per~(\ref{eq:weight}). Accordingly, no  set   containing $e$ can be part of  $\S^{(1)}_{slack}$. This leads to a contradiction, for we  assumed that $e \in \E^{(1)}_{slack}$, which in turn means that  $e$ is exclusively covered by the sets from $\S^{(1)}_{slack}$.
\end{proof}

\begin{corollary}
	\label{app:cor:prop:hyper-edge}
	We have $\E^{(1)}_{slack} = \emptyset$.
\end{corollary}

\begin{proof} Follows from the proof of Claim~\ref{app:prop:hyper-edge}.
\end{proof}

\begin{claim}
	\label{app:prop:node:weight}
	For every set $s \in \S$ we have:
	\begin{equation*}
	W(s) \in  \begin{cases}
	[(1+\epsilon)^{-1} c_s, c_s] & \text{ if } \ell(s) > 0; \\
	[0,(1+\epsilon)^{-1} c_s) & \text{ else if } \ell(s) = 0.
	\end{cases}
	\end{equation*}
\end{claim}

\begin{proof}(Sketch)
	Fix any set $s \in \S$. Initially, in the beginning of round $L$, we have $W(s) = \sum_{e \in s} w(e) \leq |\E| \cdot (Cn)^{-1} = 1/C$. From~(\ref{eq:weight}), we conclude that $0 \leq W(s) \leq c_s$ at this point in time. Subsequently, in each round $t \in \{L, \ldots, 1\}$, we keep moving down the set $s$ to level $t-1$ iff we have $W(s) < (1+\epsilon)^{-1}c_s$ in the beginning of the current round. Consider such a round $t$ where the set $s$ gets moved down to level $t-1$. In the beginning of round $t$, we had $W(s) < (1+\epsilon)^{-1}c_s$. During round $t$, we only increase (by a $(1+\epsilon)$ factor) the weights $w(e)$ of some of the elements $e$ contained in $s$. Thus, even at the end of round $t$, we have $W(s) \leq c_s$. This is sufficient for us to conclude that at the end of the algorithm, we have:  
	$$\text{(a) } W_s \leq c_s \text{ and  (b) } W_s \in [(1+\epsilon)^{-1} c_s, c_s] \text{ if } \ell(s) > 0.$$ 
	
	It remains now to consider the case where $\ell(s) = 0$ at the end of the algorithm. If this is the case, then we must have had $s \in \S^{(1)}_{slack}$, for only the sets in $\S^{(1)}_{slack}$ gets demoted to level $0$ during round $1$. Accordingly, we infer that $W_s < (1+\epsilon)^{-1} c_s$ in the beginning of round $1$. Now, by Corollary~\ref{app:cor:prop:hyper-edge}, we have $\E^{(1)}_{slack} = \emptyset$. In other words, no element changes its weight during round $1$. So the weight $W(s)$ of the set $s$ also remains unchanged during round $1$. We accordingly infer that $W(s) < (1+\epsilon)^{-1} c_s$ at the end of the algorithm.
\end{proof}

Property~\ref{prop:element:weight:level} follows from Claim~\ref{app:prop:hyper-edge:weight}. Property~\ref{prop:set:weight:level} follows from Claim~\ref{app:prop:node:weight}. Finally, Property~\ref{prop:element:level} follows from Claim~\ref{app:prop:hyper-edge} and Claim~\ref{app:prop:node:weight}.

\end{document}